\PassOptionsToPackage{table,xcdraw}{xcolor}
\documentclass[sigplan]{acmart}

\usepackage{algorithm}
\usepackage{algpseudocode}
\usepackage{graphicx}
\usepackage{subcaption}
\usepackage{multirow}
\usepackage[table,xcdraw]{xcolor}
\usepackage{commands}
\usepackage{adjustbox}
\usepackage{balance}

\usepackage{tikz}
\usetikzlibrary{matrix,decorations.pathreplacing}

\author{Keren Zhou}
\authornote{The authors contributed equally to this research.}
\email{kzhou6@gmu.edu}
\affiliation{%
  \institution{George Mason University}
  \city{Fairfax}
  \country{United States}
}
\affiliation{%
  \institution{OpenAI}
  \city{San Fancisco}
  \country{United States}
}

\author{Mario Lezcano-Casado}
\authornotemark[1]
\email{lezcano@openai.com}
\affiliation{%
  \institution{OpenAI}
  \city{San Fancisco}
  \country{United States}
}

\author{Adam P. Goucher}
\authornotemark[1]
\email{apgoucher@openai.com}
\affiliation{%
  \institution{OpenAI}
  \city{San Fancisco}
  \country{United States}
}

\author{Akhmed Rakhmati}
\email{arakhmati@openai.com}
\affiliation{%
  \institution{OpenAI}
  \city{San Fancisco}
  \country{United States}
}

\author{Jeff Niu}
\email{jeffniu@openai.com}
\affiliation{%
  \institution{OpenAI}
  \city{San Fancisco}
  \country{United States}
}

\author{Justin Lebar}
\email{jlebar@openai.com}
\affiliation{%
  \institution{OpenAI}
  \city{San Fancisco}
  \country{United States}
}

\author{Pawel Szczerbuk}
\email{pawel.szczerbuk@openai.com}
\affiliation{%
  \institution{OpenAI}
  \city{San Fancisco}
  \country{United States}
}

\author{Peter Bell}
\email{peterbell10@openai.com}
\affiliation{%
  \institution{OpenAI}
  \city{San Fancisco}
  \country{United States}
}

\author{Phil Tillet}
\email{phil@openai.com}
\affiliation{%
  \institution{OpenAI}
  \city{San Fancisco}
  \country{United States}
}

\author{Thomas Raoux}
\email{thomas.raoux@openai.com}
\affiliation{%
  \institution{OpenAI}
  \city{San Fancisco}
  \country{United States}
}

\author{Zahi Moudallal}
\email{zahi@openai.com}
\affiliation{%
  \institution{OpenAI}
  \city{San Fancisco}
  \country{United States}
}

\renewcommand{\shortauthors}{Keren Zhou et al.}

\copyrightyear{2026}
\acmYear{2026}
\setcopyright{cc}
\setcctype{by}
\acmConference[ASPLOS '26]{Proceedings of the 31st ACM International Conference on Architectural Support for Programming Languages and Operating Systems, Volume 1}{March 22--26, 2026}{Pittsburgh, PA, USA}
\acmBooktitle{Proceedings of the 31st ACM International Conference on Architectural Support for Programming Languages and Operating Systems, Volume 1 (ASPLOS '26), March 22--26, 2026, Pittsburgh, PA, USA}
\acmPrice{}
\acmDOI{10.1145/3760250.3762221}
\acmISBN{979-8-4007-2165-6/26/03}

\begin{document}
\title{Linear Layouts: Robust Code Generation of Efficient Tensor Computation Using \texorpdfstring{$\F$}{F(2)}}

\begin{abstract}
Efficient tensor computation is a cornerstone of modern deep learning (DL) workloads, yet existing approaches struggle to achieve flexible and performant design and implementation of tensor layouts—mappings between logical tensors and hardware resources.  
The increasing complexity of DL algorithms and hardware demands a generic and systematic approach to handling tensor layouts.  
In this work, we introduce \emph{Linear Layouts}, a novel approach that models tensor layouts using linear algebra over $\F$.  
By representing tensor layouts as binary matrices acting on the bits of the hardware representation, our approach enables a generic layout definition—as opposed to the classical case-by-case approach—and allows for generic layout-to-layout conversions, eliminating the quadratic explosion that plagues existing solutions.
We integrate linear layouts with Triton and demonstrate their effectiveness in optimizing individual Triton operators as well as kernels written in Triton.  
We also show that linear layouts reduce engineering effort in the compiler backend while fixing several bugs in Triton's legacy layout system.
\end{abstract}


\maketitle 


\section{Introduction}

Deep learning (DL) models are rapidly growing in both scale and architectural complexity~\cite{shoeybi2019megatron,rajbhandari2020zero}.
Modern DL models such as deep transformers now contain billions of parameters~\cite{dubey2024llama,achiam2023gpt} and employ varied structures~\cite{fedus2022switch,vaswani2017attention,kwon2023efficient} with low precisions~\cite{frantar2022gptq,liu2024spinquant,lin2024awq}, pushing the limits of current hardware and software optimizations.
Notably, there is a pressing need for more efficient tensor computation~\cite{abadi2016tensorflow,jax2018github,ansel2024pytorch}, which is a fundamental building block of DL models.
The performance of tensor computation often relies on sophisticated mappings between logical tensors and hardware compute and memory resources, which we refer to as \textit{tensor layouts}~\cite{hidet202X, graphene202X, amos202X}.
We demonstrate two example layouts in~\cref{fig:example-layout}.

The growing complexity of DL hardware, such as GPUs, leads to increasingly intricate tensor layouts.
For example, to enable efficient matrix multiplication, Nvidia GPUs incorporate different layouts to use Tensor Cores on Ampere, Hopper, and Blackwell generations, each with different variants when using different data types~\cite{nvidia_ptx_isa}.
Other GPU vendors, such as AMD and Intel, implement distinct layouts when leveraging their tensor core equivalence~\cite{amd_matrix_cores,intel_xmx} for acceleration.
Consequently, the rapid advancements in hardware architectures and varied DL models demand a new approach to modeling tensor layouts.

\begin{figure}[tp]
  \centering
  \begin{subfigure}[b]{0.45\linewidth}
    \centering    \includegraphics[width=\linewidth]{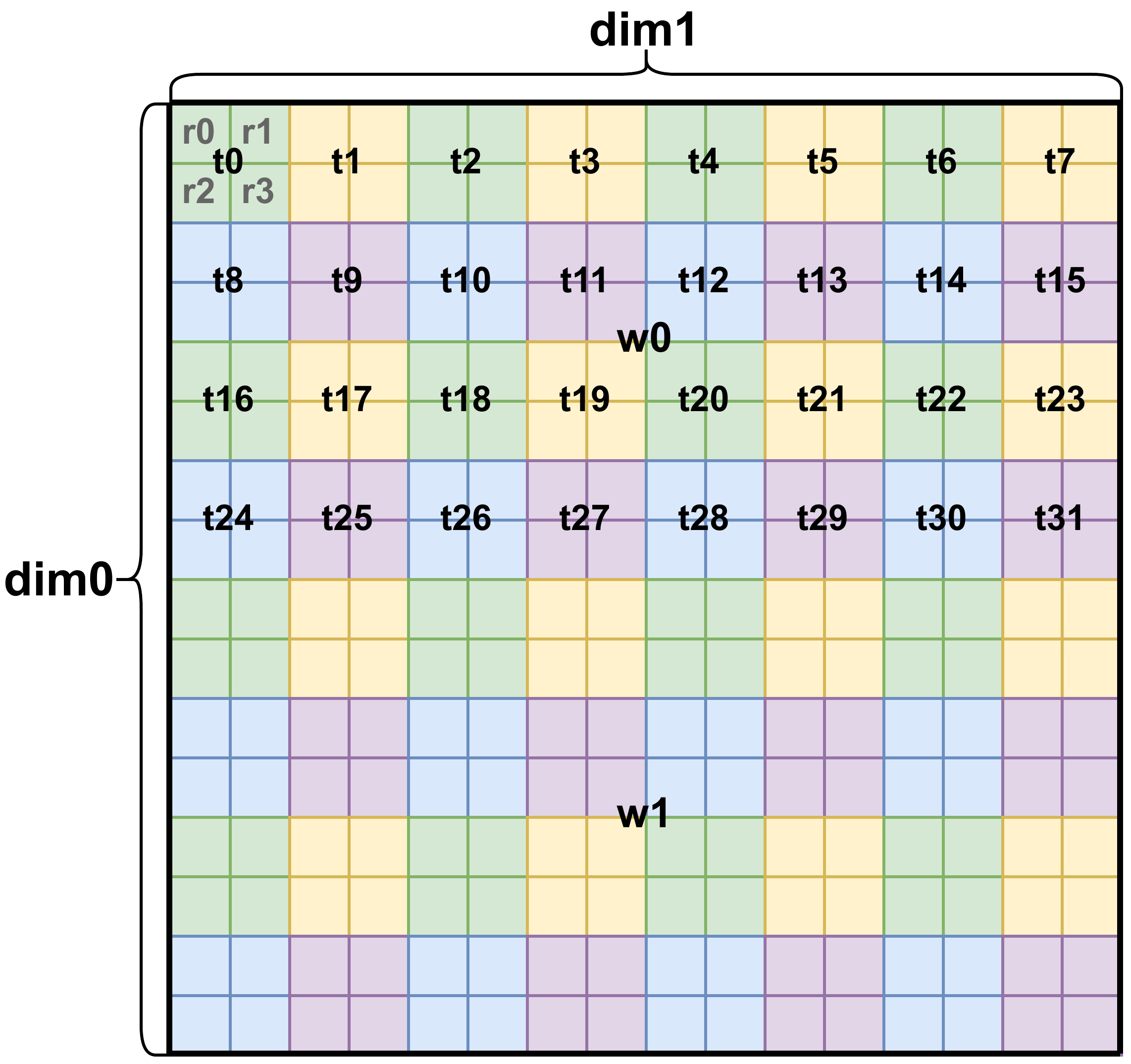}
    \caption{Layout A}
    \label{fig:example-layout-1}
  \end{subfigure}
  \hfill
  \begin{subfigure}[b]{0.45\linewidth}
    \centering
    \includegraphics[width=\linewidth]{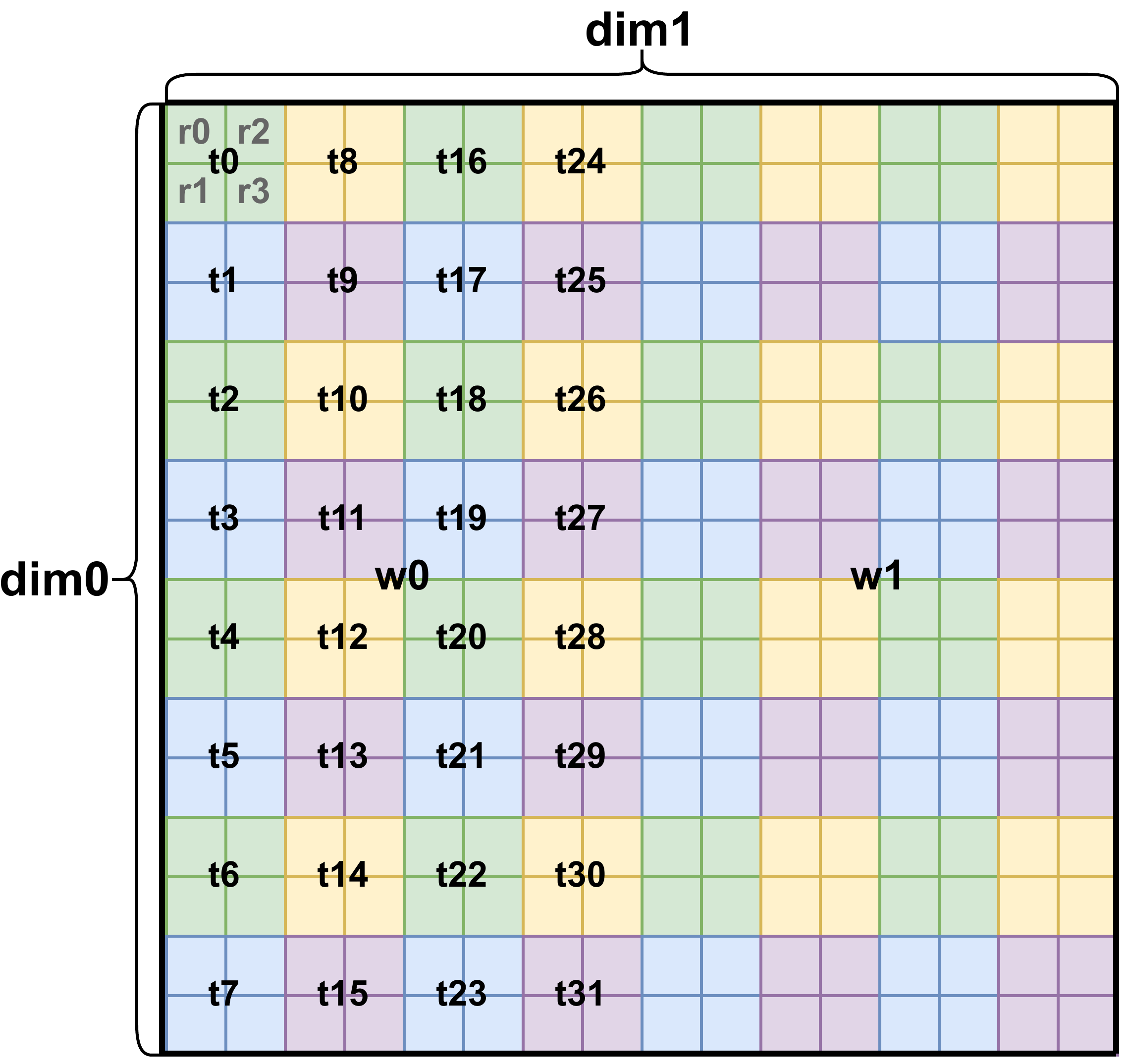}
    \caption{Layout B}
    \label{fig:example-layout-2}
  \end{subfigure}
  \caption{Two different layouts storing a $16\times16$ tensor using two warps. $w_{i}$ denotes warp $i$, $t_{j}$ denotes thread $j$, and $r_{k}$ denotes register $k$.
  If the tensor is stored in row-major format, loading it into layout A is more efficient than into layout B due to coalesced reads.}
    \vspace{-2em}
  \label{fig:example-layout}
\end{figure}

However, current DL programming models and libraries for tensor computation lack a solution for flexible and efficient tensor layout construction and conversion.
DL practitioners often rely on highly-optimized vendor libraries (e.g., NVIDIA cuDNN~\cite{chetlur2014cudnn}, cuBLAS~\cite{nvidia_cublas}) to achieve peak performance.
While these libraries excel for standard operations, they support only a limited set of tensor operators.
A custom operator introduced by a new model falls outside their coverage, forcing developers to implement GPU kernels from scratch, dealing with intricate layout-related issues.
DL compilers such as TVM~\cite{TVM2018}, XLA~\cite{XLA2020}, and Triton~\cite{tillet2019triton} implement tensor layouts as special attributes within their compiler backends.
However, only a limited set of layouts and conversions between layouts are supported in these compilers, lacking a generic, robust, and systematic framework.
Defining custom layouts requires substantial modifications to the compiler, leading to a quadratic explosion of layout-to-layout conversions.
Manually implementing these layouts and their conversions is often error-prone; to date, 12\% of bugs filed in Triton's GitHub repository~\cite{triton_layout_bugs} are layout-related.
Moreover, without treating tensor layouts as a first-class citizen for optimization, often suboptimal data movement incurs in tensor computation and layout conversions.
For example, FlashAttention 3~\cite{shah2025flashattention} manually optimizes data movement using byte permute and warp shuffle instructions to bypass shared memory in layout conversions—an approach that has not yet been implemented in DL compilers.

Bridging this gap requires overcoming several technical challenges.
First, we need a general and composable representation for mapping tensors to hardware resources.
Second, layout conversions should be expressed within a unified formalism, incorporating even complex transformations such as data swizzling~\cite{SwizzlingComputerGraphics}.
Third, this representation must seamlessly integrate with low-level hardware optimizations to ensure efficient data access and computation.

In this work, we propose \textit{Linear Layouts}, an approach that addresses these challenges by treating tensor layouts as linear mappings between vector spaces over the field $\F$, leveraging linear algebra as a unifying abstraction for operations on layouts.
Every tensor layout is modeled as a linear function---a matrix---that maps physical resource indices into a logical tensor of size $2^n$ using binary arithmetic on the bits of the input and the output.
This way, complex representations such as swizzling and broadcasting are naturally expressed as combinations of $\operatorname{XOR}$ and $\operatorname{AND}$ operations on bit-vectors.
Furthermore, arbitrary layout conversions can be composed using matrix transformations such as matrix multiplication and inverse, which enable a formal characterization of data exchanges both across and within the hardware hierarchy, thereby allowing the compiler to generate efficient hardware primitives for data movement generically.
It eliminates the need for hard-coded, case-by-case handling of layouts---any layout that can be represented as a permutation of indices or via swizzling can be plugged into our framework and automatically optimized.

\begin{figure}
    \centering
    \includegraphics[width=0.8\linewidth]{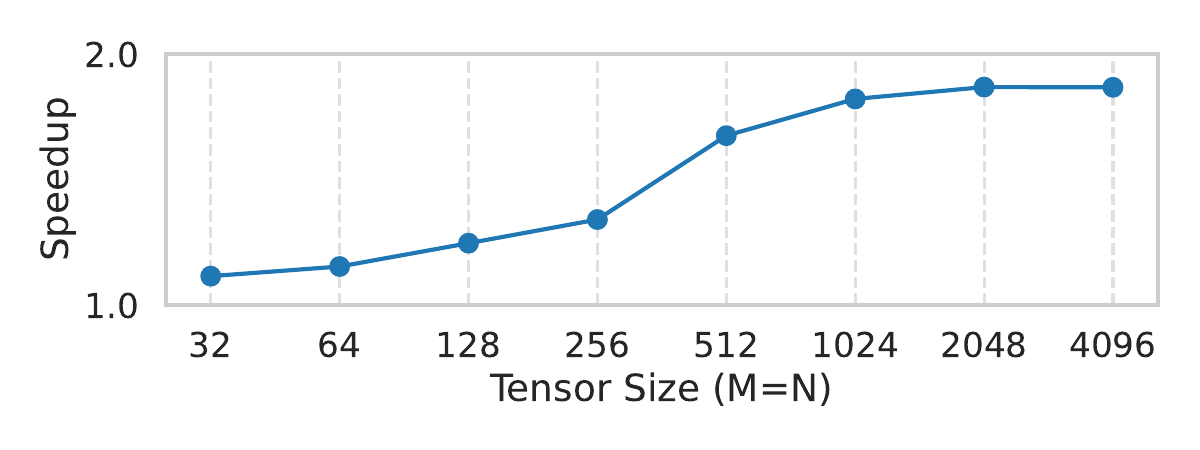}
    \caption{Speedups across a range of tensor shapes compared to the padding heuristic in a float8 transpose kernel for tensors of size $M \times N$.}
    \label{fig:matrix-size}
    \vspace{-1em}
\end{figure}

We implement linear layouts as part of the code generation workflow in Triton’s GPU backend, which is widely used to customize deep learning operators on GPUs from various vendors.
To assess the effectiveness of linear layouts, we compare the correctness and performance of the generated kernels against those produced by legacy Triton, which does not use linear layouts.
Legacy Triton relies on heuristics (e.g., shared memory padding) for layout-based code generation and optimization, which are effective for common access patterns.
However, we observe that it causes many bugs in layout conversions, lacks extensibility for supporting flexible layouts, and delivers suboptimal performance for complex tensor access patterns (see \autoref{fig:matrix-size}).
Experimental results demonstrate that our approach improves correctness and yields up to 1.40$\times$ speedup, with an average of 1.07$\times$ across 265 real-world benchmark cases.
This paper makes the following contributions:
\begin{itemize}
    \item We present linear layouts, a novel approach that uses linear algebra over $\F$ to represent and compose tensor layouts within a unified framework.
    \item We fully integrate linear layouts into Triton's GPU backend, implementing a layout engine that is able to automatically choose and propagate layouts for any operation in Triton.
    \item We introduce novel algorithms, including automatic optimal swizzling discovery that provably maximizes read/write vectorization and minimizes read/write bank conflicts, automatic optimal warp-shuffle generation, and generic lowering of hardware intrinsics for all the layouts of this family.
    \item We evaluate linear layouts on both synthetic and real DL workloads, demonstrating that it outperforms existing baselines.
    Furthermore, we demonstrate that linear layouts enhances robustness by fixing many pre-existing bugs in Triton.
\end{itemize}

\section{Background}
In this section, we introduce the architecture and mathematic background necessary for this paper.

\subsection{GPU Architectures}

Modern GPUs are designed to exploit extreme parallelism through a hierarchical execution model that includes multiple levels of hardware resources.
The key execution units include cooperative thread arrays (CTAs), warps, and threads.
Each GPU thread has access to private registers, which offer the lowest-latency storage but are limited in capacity.
Regular instructions can be executed independently by individual threads.
However, some special function units must be executed at a higher granularity level.
For example, NVIDIA's \texttt{mma} (matrix multiply-accumulate) instruction~\cite{nvidia_ptx_isa} utilizes tensor cores by performing multiple multiply–add operations in parallel, issued by individual warps.
Advanced variants such as \texttt{wgmma} (warp group matrix multiply-accumulate)~\cite{nvidia_ptx_isa} extend these capabilities by executing matrix multiplication on multiple warps together.
AMD has also introduced similar primitives, such as \texttt{mfma} (matrix fused multiply–add) instructions~\cite{amd_matrix_cores}.
Note that these instructions require data to be distributed across threads and warps, or reside in shared memory or special memory units (e.g., Tensor Memory on Blackwell~\cite{nvidia_ptx_tensor_memory}) in special layouts to yield correct results.
However, these layouts do not typically yield the best performance for other operations like load/store, and not always can one use specific instructions to directly copy data from the global memory to the special memory units.
As a result, data must often be rearranged so that the layout used for memory accesses (which emphasizes coalescence and bandwidth) is converted into the layout preferred by the compute units (which emphasizes arithmetic throughput).
In short, achieving peak performance requires not only leveraging these specialized units but also carefully designing tensor layouts and conversions.

\subsection{Triton Language and Compiler}

Triton~\cite{tillet2019triton} is a Python-like domain-specific language designed to offer flexible interfaces for writing high-performance deep learning primitives.
Triton's compiler backend leverages MLIR~\cite{lattner2021mlir}, which enables the expression of abstractions at multiple levels and facilitates a smooth lowering process through a series of dialects.

At its core, a Triton kernel follows the single program multiple data (SPMD) model, wherein computation is partitioned into multiple abstract Triton program instances.  
This design allows developers to focus primarily on parallelism at the CTA level, as an operator in Triton is applied across all threads within each program instance.
In Triton, the term \emph{tensor} refers to \emph{tiles} extracted from the original PyTorch tensors, which serve as the inputs and outputs for GPU kernels.  

During compilation, Triton's Python code is first translated into the Triton dialect (\code{tt}), which is further translated into the TritonGPU dialect (\code{ttg}).  
Throughout this process, each tensor is associated with a specific layout to take full advantage of hardware function units available on modern GPUs.
For instance, Tensor Cores and similar units are utilized with a \code{mma} layout when dot-like operators~\cite{triton_language} (e.g., \code{tt.dot} and \code{tt.dot\_scaled}) are encountered.


\subsection{Linear Algebra Preliminaries}
We introduce the following concepts that provide the foundation for the Linear-Layout transformations used in subsequent sections.

\begin{itemize}
  \item \textbf{Vector Space.}  
        Let $\mathbb{F}$ be a field (e.g., $\mathbb{R}$).  
        A \emph{vector space} is a non-empty set $V$ (e.g., $\mathbb{R}^3$) of vectors over $\mathbb{F}$ and is equipped with  
        vector addition and scalar multiplication satisfying the eight vector-space axioms (associativity, commutativity, identity, inverses for addition; distributivity, compatibility, identity for scalar multiplication).

  \item \textbf{Subspace.}  
        A non-empty subset $S \subseteq V$ is a \emph{subspace} of $V$ if it is closed under the inherited operations.

  \item \textbf{Linear Combination.}  
        Given a set of $x_{1},x_{2},\dots,x_{n}\in V$ and scalars $a_{1},a_{2},\dots,a_{n}\in\mathbb{F}$,  
        the vector $v$ is a \emph{linear combination} of the vectors $x_{1},\dots,x_{n}$ if $v$ can be written in the following form.
        \[
            v = a_{1}x_{1} + a_{2}x_{2} + \dots + a_{n}x_{n}
        \]

  \item \textbf{Linear Independence.}
        A set of vectors $x_{1},\dots,x_{n} \in V$ is \emph{linearly independent} if the following equation has no nontrivial solutions $(a_1, \dots, a_n) \neq (0, \dots, 0)$:
        \[
            a_{1}x_{1} + a_{2}x_{2} + \dots + a_{n}x_{n} = 0
        \]
        
  \item \textbf{Span.}  
        For a subset $S\subseteq V$, the \emph{span} of $S$ is the set of all linear combinations of vectors in $S$:
        \[
            \sspan(S)\;=\;
            \left\{\sum_{i=1}^{k} a_{i} s_{i}\,|\,
            k\in\mathbb{N},\; s_{i}\in S,\; a_{i}\in\mathbb{F}\right\}.
        \]
        It is the smallest subspace of $V$ containing $S$.
  
  \item \textbf{Basis.}
        For a subset $S\subseteq V$, a \emph{basis} of $S$ is a set of vectors $x_{1},\dots,x_{n}$ such that $S=\sspan(\{x_{1},\dots,x_{n}\})$ and the set $\{x_{1},\dots,x_{n}\}$ is linearly independent.
\end{itemize}

\subsection{$\F$ Mathematics}

We denote the field of two elements $\{ 0, 1\}$ as $\mathbb{F}_2$.
In $\mathbb{F}_2$, all arithmetic operations are performed modulo 2. 
For example, addition is defined by
\[
a \oplus b = (a + b) \bmod 2 = a \operatorname{XOR} b
\]
which corresponds to logical $\operatorname{XOR}$, while multiplication is given by
\[
a \cdot b = (a \times b) \bmod 2 = a \operatorname{AND} b
\]
corresponding to logical $\operatorname{AND}$.

An essential operation in linear algebra over $\mathbb{F}_2$ is matrix multiplication. Let
\[
A \in \mathbb{F}_2^{m \times n} \quad \text{and} \quad B \in \mathbb{F}_2^{n \times p}
\]
be matrices whose elements are in $\mathbb{F}_2$.
The product $C = AB \in \mathbb{F}_2^{m \times p}$ is defined element-wise by
\[
c_{ij} = \bigoplus_{k=1}^{n} a_{ik} \cdot b_{kj},
\]
where the summation $\bigoplus$ represents repeated addition in $\mathbb{F}_2$ (i.e., XORing the products $a_{ik} \cdot b_{kj}$). This is analogous to standard matrix multiplication, with the distinction that all arithmetic is performed in $\F$.

Arithmetic in $\F$ naturally aligns with binary logic, making operations in this field highly efficient in hardware implementations. Consequently, $\F$ is widely used in areas such as cryptography~\cite{mullen2002finite} and error-correcting codes~\cite{pretzel1992error}.

\section{Overview}
\cref{fig:triton-layouts} lists all layouts available in Triton.
At the highest level, layouts are divided into Distributed and Memory layouts, where the former indicates that tensor elements are ``distributed'' across different execution units, while the latter indicates that tensor elements are stored in certain special memory.
Distributed layouts are further classified into types, including Blocked, Sliced, MMA, and MMA Input layouts, while Memory layouts can be further classified into Unswizzled and Swizzled layouts.
Blocked layouts are often used for contiguous memory accesses. 
MMA and MMA input layouts are used for the output and inputs of matrix multiplication operations (e.g., \code{tt.dot}).
MMA layouts can be further classified according to hardware instructions they map to, such as \code{mma} and \code{wgmma} on NVIDIA GPUs, or \code{mfma} on AMD GPUs.
Sliced layouts extract a dimension from their parent layout, used as the input to a broadcast or the output of a reduction.

\begin{figure}[tp]
    \centering
    \includegraphics[width=1.0\linewidth]{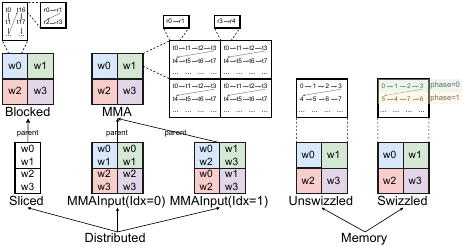}
    \caption{Legacy layouts in Triton. $w_{i}$ denotes warp $i$, $t_{j}$ denotes thread $j$, and $r_{k}$ denotes register $k$.}
    \label{fig:triton-layouts}
\end{figure}

The legacy Triton layout system requires each layout to define its own interface methods—such as the number of elements per thread and the number of contiguous elements. 
Moreover, indexing into tensor elements, as well as conversions between layouts, must be explicitly implemented for each layout.
This approach resulted in buggy layout constructions and conversions~\cite{triton_layout_bugs}.

\paragraph{Our Approach}
In contrast, our approach defines layouts using a linear layout-based mechanism.  
For backward compatibility, we also provide utilities to define each legacy layout as a linear layout.  
Once a layout is defined using these utilities, interface methods such as \code{getNumElementsPerThread} no longer need to be reimplemented.  
With this approach, arbitrary layouts can be instantiated without modifying the core Triton compiler backend, including those for out-of-tree backends such as Intel GPUs.
Additionally, our approach automatically enables robust conversion between layouts and unifies the determination of hardware resources in code generation.


\section{Linear Layouts}
This section covers the definition of linear layouts, some fundamental linear layout operators, the creation of various Triton layouts as instances of linear layouts, and a general layout engine applied to Triton.
Proofs of propositions presented in this section are provided in the Appendix unless stated otherwise.

\subsection{A Motivating Example}\label{sec:ll_intro}
Most parameters in GPU programming are powers of two: a warp consists of $32$ or $64$ threads, a warp group contains $4$ warps, and matrix multiplication intrinsics (e.g., \code{mma} and \code{wgmma}) require tile dimensions of size $16\times n$, where $n \geq 1$.
Further, in Triton's programming model, the dimensions of tensors, as well as subdivisions of layouts associated with each tensor, such as the registers per thread and the number of threads, are restricted to powers of two.
In \cref{fig:example-layout}, layout A tiles a $16\times 16$ tensor using $2\times2$ registers, $4\times8$ threads, and $2\times1$ warps.

\begin{table}[tp]
\centering
\scriptsize
\caption{Some elements from the top-left corner of Layout A (Figure~\ref{fig:example-layout-1}). We show the mapping from matrix locations to register, thread, and warp in binary representations.}
\label{tab:mapping}
\begin{adjustbox}{max width=\columnwidth}
\begin{tabular}{@{}cccc@{}}
\toprule
\textbf{Location} & \textbf{Register} & \textbf{Thread} & \textbf{Warp} \\
\midrule
(0, 0) / (0b0, 0b0) & $r_{0}$ / 0b0 & $t_{0}$ / 0b0 & $w_{0}$ / 0b0\\
(0, 1) / (0b0, 0b1) & $r_{1}$ / 0b1 & $t_{0}$ / 0b0 & $w_{0}$ / 0b0 \\
(0, 2) / (0b0, 0b10) & $r_{0}$ / 0b0 & $t_{1}$ / 0b1 & $w_{0}$ / 0b0 \\
(0, 3) / (0b0, 0b11) & $r_{1}$ / 0b1 & $t_{1}$ / 0b1 & $w_{0}$ / 0b0 \\
\dots              & \dots   & \dots   \\
(1, 0) / (0b1, 0b0) & $r_{2}$ / 0b10 & $t_{0}$ / 0b0 & $w_{0}$ / 0b0 \\
(1, 1) / (0b1, 0b1) & $r_{3}$ / 0b11 & $t_{0}$ / 0b0 & $w_{0}$ / 0b0 \\
\dots              & \dots   & \dots   \\
(2, 2) / (0b10, 0b10) & $r_{0}$ / 0b0 & $t_{9}$ / 0b1001 & $w_{0}$ / 0b0 \\
(2, 3) / (0b10, 0b11) & $r_{1}$ / 0b1 & $t_{9}$ / 0b1001 & $w_{0}$ / 0b0 \\
\dots              & \dots   & \dots   \\
(3, 2) / (0b11, 0b10) & $r_{2}$ / 0b10 & $t_{9}$ / 0b1001 & $w_{0}$ / 0b0 \\
(3, 3) / (0b11, 0b11) & $r_{3}$ / 0b11 & $t_{9}$ / 0b1001 & $w_{0}$ / 0b0 \\
\dots              & \dots   & \dots   \\
\bottomrule
\end{tabular}
\end{adjustbox}
\end{table}

Because these quantities are powers of two, visualizing the distribution of elements in layout A (as shown in~\cref{tab:mapping}) is straightforward using the bit representation of their coordinates.
Register 0 ($r_{0}$) of all threads is located at coordinates $(i, j)$, where the last bits of both $i$ and $j$ are $0$.
For example, $r_{0}$ of thread $t_{1}$ is located at $(0, 2)=(0b0, 0b10)$.
For comparison, $r_{1}$ elements have coordinates where the last bit of $i$ is always $0$, while the last bit of $j$ is always $1$.
For example, $r_{1}$ of $t_{9}$ is located at $(2, 3)=(0b10, 0b11)$.
This is because each thread takes a $2\times2$ tile consecutively in the logical tensor.

More generally, we can consider three mapping functions that let us express every hardware index as a coordinate inside ever-larger tiles, including register $\rightarrow \text{loc}_{\text{thread}}$, thread  $\rightarrow \text{loc}_{\text{register}}$, and warp $\rightarrow \text{loc}_{\text{warp}}$.  
As an example, take register $r_{1}$ in thread $t_{9}$ of warp $w_{0}$.  
The register mapping places $r_{1}$ at $\text{loc}_{r_{1}} = (0, 1) = (0b0, 0b1)$, the second column of the $t_{9}$ tile.
The thread map situates that tile at $\text{loc}_{t_{9}} = (2, 2) = (0b10, 0b10)$ within the warp tile, and the warp map assigns to $\text{loc}_{w_{0}} = (0, 0)$. 
Bitwise XOR of these three coordinate pairs yields the register’s absolute position, $\text{loc}_{w_{0}} \oplus \text{loc}_{t_{9}} \oplus \text{loc}_{r_{1}}=(2, 3) = (0b10, 0b11)$.

Putting all this together, if we consider a vector $v$ of size $8$ represents an element of a thread in a warp, where the first $2$ bits $v_{0:1}$ represent the register ($\reg$), the next $5$ bits $v_{2:6}$ represent the thread ($\thread$), and the last bit $v_7$ represents the warp ($\warp$), we can define layout $A=\F^{8\times8}$.

\[
\footnotesize
A = \left[ \begin{array}{c|cc|ccccc|c}
  \multicolumn{1}{c}{ } & \multicolumn{2}{c}{\reg} & \multicolumn{5}{c}{\thread} & \multicolumn{1}{c}{\warp} \\
  \hline
  \multirow{4}{*}{\text{j}} & 1 & 0 & 0 & 0 & 0 & 0 & 0 & 0 \\
 & 0 & 0 & 1 & 0 & 0 & 0 & 0 & 0 \\
 & 0 & 0 & 0 & 1 & 0 & 0 & 0 & 0 \\
 & 0 & 0 & 0 & 0 & 1 & 0 & 0 & 0 \\
 \hline
\multirow{4}{*}{\text{i}} & 0 & 1 & 0 & 0 & 0 & 0 & 0 & 0 \\
 & 0 & 0 & 0 & 0 & 0 & 1 & 0 & 0 \\
 & 0 & 0 & 0 & 0 & 0 & 0 & 1 & 0 \\
 & 0 & 0 & 0 & 0 & 0 & 0 & 0 & 1
\end{array} \right]
\]

We can obtain $v$'s location $(i, j)$ in the tensor through $w = Av \in \mathbb{F}_2^8$, where $w_{0:3} = j$ and $w_{4:7} = i$, given that $j$ is the fastest moving dimension.

\paragraph{Labeled Vector Spaces}
We assign labels to each bit in the layout.
The input $v$ resides in $\F^2 \times \F^5 \times \F^1$, modeling the space of $\reg \times \thread \times \warp$.
The output $w$ follows an $\F^4 \times \F^4$ structure, representing the two dimensions of the logical tensor $(i, j)$.

To better understand the location calculation using matrix vector multiplication, consider register $r_{1}$ in thread $t_{9}$ of warp $w_{0}$, where $v_{\reg} = 0b01 = [1\,0]^{T}$\footnote{The least significant bits come first in the vector}, $v_{\thread} = 0b01001 = [1\,0\,0\,1\,0]^{T}$, and $v_{\warp} = 0b0 = [0]^{T}$.
Conducting $Av$ will XOR the bitwise product for each row of $A$ with $v$ and yield $w_{j} = [1\,1\,0\,0]^{T} = 0b0011=3$ and $w_{i} = [0\,1\,0\,0]^{T} = 0b0010 = 2$.

\subsection{Definition and Constructions}


\begin{definition}[Linear Layouts]
    We define a \textbf{Linear Layout} as a linear map between (labeled) vector spaces over $\F$.
\end{definition}

For example, we can define layout $L$ as $\deffun{L: \reg \times \thread \times \warp -> \F^n \times \F^m;}$, and we denote each labeled subspace of $L$ using a subscript, such as $ L_{\reg}$.
In the next, we review basic linear algebra over $\F$ to construct specialized layouts.

\begin{definition}[Composition]
    Given vector spaces $U, V, W$ over $\F$ and linear layouts $\deffun{L_1 : U -> V;}$ and $\deffun{L_2 : V -> W;}$, we define their composition as
\[
    \deffun{L_2 \circ L_1 : U -> W;u -> L_2\pa{L_1\pa{u}}}
\]
Representing $L_{1}$ and $L_{2}$ as matrices $M_1$ and $M_2$, the matrix representing $L_2 \circ L_1$ is given by the (label-wise) matrix multiplication $M_2M_1$ over $\F$.
\end{definition}

\begin{definition}[Product]\label{def:product}
    Given two vector spaces $U, V$ over $\F$, we define their product as
\[
    U \times V = \set{(u, v) | u \in U, v \in V}.
\]
Given two linear layouts $\deffun{L_1 : U_1 -> V_1;}$, $\deffun{L_2 : U_2 -> V_2;}$, and $u_1 \in U_1, u_2 \in U_2$, we define their product \footnote{This construction is more often known as the direct sum of maps $L_1 \oplus L_2$. We choose to discuss it as the categorical product to avoid creating confusion with the notation for the $\operatorname{XOR}$.} as
\[
    \deffun{L_1 \times L_2 : U_1 \times U_2 -> V_1 \times V_2;(u_1, u_2) -> \pa{L_1\pa{u_1}, L_2\pa{u_2}}}
\]
Representing $L_{1}$ and $L_{2}$ as matrices $M_1$ and $M_2$, the matrix representing $L_1 \times L_2$ is given by the (label-wise) block-diagonal matrix
\[
\begin{bmatrix}
    M_1 & 0 \\
    0   & M_2
\end{bmatrix}.
\]
\end{definition}

\emph{Composition} and \emph{Product} operations are used to combine simple layouts into more complex ones.
For example, composition can extract a slice from the parent layout by mapping one of the parent dimensions to all zeros.
The product operation can be used to incrementally construct a complex layout, progressing from registers to threads to warps.
We also define the inverse operation of product (when it exists) in the following.

\begin{definition}[Left Division]\label{def:ldivision}
	A matrix $M$ is divisible on the left by a matrix $M_1$ if $M$ has the structure
\[
M = \begin{bmatrix}
    M_1 & 0 \\
    0   & M_2
\end{bmatrix}.
\]
We denote the division on the left as $M \ldiv M_1 = M_2$.
We handle this operation label-wise in a linear layout.
\end{definition}
Left division can be useful for determining whether a layout can be decomposed into smaller layouts that satisfy efficient hardware primitives, such as \code{ldmatrix}, as further discussed in \cref{sec:hw_primitives}.

\begin{definition}[Right Inverse]
    A surjective linear layout $\deffun{L : U -> V;}$ over $\F$ has a right inverse.

    If $M$ is a matrix representation of $L$ of shape $m \times n$ we define $M^{-1}$ as the $n \times m$ least squares solution of $MX = I_m$ where $I_m$ is the $m \times m$ identity matrix. In particular, it can be computed via Gaussian elimination over $\F$.
\end{definition}

Inversion is used when one needs to recover  hardware indices from coordinates in the logical tensor.

\subsection{Completeness}
We discussed the example in~\cref{sec:ll_intro} how layout A in~\cref{fig:example-layout} forms a linear layout.
We can easily generalize this family of layouts by using the concepts presented in the previous section.
This family of layouts is referred to as the \textbf{Blocked Layouts} in the legacy Triton layout system.

\begin{proposition}\label{prop:blocked}
    Blocked layouts are linear layouts.
\end{proposition}

Blocked layouts are one kind of \textbf{Distributed Layouts} in Triton, which is referred to as any layout that is used to describe distribution on registers, threads, and warps.
We label their dimensions as $\reg, \thread, \warp$.
Other commonly used distributed layouts are the ones associated with matrix multiplication operations like \code{mma} and \code{wgmma} operations on NVIDIA GPUs.
Similarly, it is possible to constructively show that layouts for AMD and Intel's matrix multiplication intrinsics exist.
We refer to the input and output of these instructions as the family of \textbf{MMA Layouts}.

\begin{proposition}
    The input and output layouts of \code{mma} and \code{wgmma} are linear layouts.
\end{proposition}

The last distributed layout is the family of \textbf{Sliced Layouts} defined as the result of applying a reduction operation (\code{tt.sum}, \code{tt.min}\textellipsis) along a dimension.
\begin{proposition}
    The slice of a linear layout is a linear layout
\end{proposition}
\begin{proof}
    Removing a dimension is a linear map.
\end{proof}

\textbf{Remark}.
    When representing the layout as a matrix, a sliced layout removes some rows of it.
    As such, the resulting layout may not be injective (some of its columns may be zero), but it will be surjective.

\begin{theorem}
    Every distributed layout is a linear layout.
\end{theorem}

We can now establish the following formal definition of distributed layouts using linear layouts.

\begin{definition}[Distributed Layout]~\label{thm:distr}
    A distributed layout in Triton is a surjective linear layout from registers, threads, and warps into a logical tensor where each column of the associated matrix has at most one non-zero bit, and no two non-zero columns are repeated.
\end{definition}

In other words, a distributed layout is a permutation matrix that may have additional zero columns interleaved.
This characterization is notably significant, as now we have fully translated into linear algebra and code what previously was specified as informal definitions.

The other family of layouts in Triton is \textbf{Memory Layouts}.
A memory layout is a way to distribute a logical tensor on a programmable segment of memory (\eg, shared memory, tensor memory, \etc).
We model it as a map from memory offsets $\offset$ to coordinates in the logical tensor.
The simplest memory layout is \textbf{Unswizzled Layouts}, which maps memory offsets directly to a logical tensor.
That is, a memory location $(i, j)$ corresponds to the coordinates $(i, j)$ in the logical tensor. 
However, when using unswizzled layouts to read from or write to certain distributed layouts, such as those in the MMA family, performance degrades due to \textbf{bank conflicts}.
To address this issue, \textbf{\code{mma} swizzling} was introduced, enabling fast memory access when reading from or writing to MMA layouts.

\begin{definition}[\code{mma} swizzling]\label{def:swizzling}
    Given parameters $\text{vec} > 0$, $\text{per\_phase}, \text{max\_phase} \geq 0$, all of them being powers of two, we define \code{mma} swizzling as a mapping from each element's location $(i, j)$ to its offset
    \[
\pa[\big]{\pa[\big]{\tfrac{i}{\text{per\_phase}} \bmod\text{max\_phase}} \oplus \tfrac{j}{\text{vec}}} \cdot \text{vec} \oplus \pa{j\bmod\text{vec}}.
    \]
    where $\cdot$ denotes multiplication over $\code{uint64}$ and $\oplus$ denotes $\operatorname{XOR}$, and the offsets are counted in elements.
\end{definition}

We can now prove the following:

\begin{proposition}
    \code{MMA} swizzled layouts are linear layouts.
\end{proposition}
\begin{proof}
    The operations involved are linear on the bits of $i, j$, so the map is linear. It is clear that it is injective and surjective, so it has an inverse and its inverse defines a linear layout from coordinates in the logical tensor to $\offset$.
\end{proof}

Computing the inverse of the map above reveals that the matrix representation of the linear layout associated to \code{mma} swizzling for a tensor of size $2^m\times2^n$ has the structure:
\[
\begin{bmatrix}
    \I_n & C \\
    0   & \I_m
\end{bmatrix}.
\]
where $\I_{m}$ and $\I_{n}$ denote identity matrices of size $m$ and $n$ accordingly.
Each row $c_i$ in $C$ is given by
\[
c_i = \pa{\text{vec} \cdot \pa{\tfrac{2^i}{\text{per\_phase}} \bmod \text{max\_phase}}} \bmod 2^n.
\]

Similar computations for other swizzling strategies yield:
\begin{theorem}
    Every memory layout is a linear layout.
\end{theorem}

We can now formally define the family of memory layouts.
\begin{definition}[Memory Layout]\label{def:memory}
    A memory layout in Triton is an invertible linear layout where the columns of the associated matrix have either $1$ or $2$ non-zero bits.
\end{definition}

We will discuss in~\cref{sec:conversion} how to compute optimal memory layouts to maximize read and write performance for arbitrary distributed layouts.

\subsection{Closure Under Triton Operations}

Triton's operations fall into four categories: (1) computation, (2) memory (global, shared, tensor, etc.), (3) layout conversion, and (4) shape operations.
In the previous section, we discussed how linear layouts allow us to handle the first two categories.
In this section, we explore how linear layouts enable the propagation of layouts through shape operations and facilitate the movement of elements from one layout to another using layout conversion operations, leveraging a generic layout engine.

\paragraph{Triton's Layout Engine}
Initially, Triton assigns blocked layouts to global memory operations and to computation operations that require specific input layouts, such as \code{mma} or \code{wgmma} (exposed via \code{tt.dot}).
We refer to these as \emph{anchor} layouts.
The propagation phase consists of a \emph{forward} pass and a \emph{backward} pass.
During the forward pass, layouts are propagated along use chains, merging candidate layouts at operations with multiple inputs. Conflicts are resolved using a heuristic model (\eg, favoring blocked layouts for load/store operations).
At this stage, layout conversions are inserted to standardize values with multiple candidate layouts.
In the backward pass, layout conversions are rematerialized in reverse through the definition chain. If the instructions along the chain are inexpensive, the entire operation chain may be rematerialized to eliminate layout conversions.

\paragraph{Propagation Through Shape Operations}
Consider the shape operations in Triton, including \code{tt.trans}, \code{tt.reshape}, \code{tt.join}, \code{tt.split}, \code{tt.expand\_dims}, and \code{tt.broadcast}.
For every input (resp.\ output) distributed layout, there exists an output (resp.\ input) layout from the same family such that the operation effectively becomes a no-op, which is inexpensive.
We prove in the appendix that the family of distributed layouts, as defined in~\cref{thm:distr}, is forward (resp.\ backward) closed under these operations.
Note that the family in~\cref{thm:distr} contains strictly more layouts than legacy layouts.
For example, legacy layouts cannot represent the transpose of an MMA layout, whereas the characterization in~\cref{thm:distr} clearly includes it.
Consequently, with legacy layouts, it was not possible to propagate layouts for some of the operations, leading to unnecessary layout conversions (\ie, additional data movement).
Linear layouts allow this engine to be as generic as possible, enabling optimizations as sophisticated as those in~\cref{sec:mxfp} to be implemented directly in the Python frontend at zero runtime cost.

\section{Code Generation}\label{sec:codegen}

Linear layouts provide a structured foundation for developing algorithms at both the language frontend and the compiler backend.
This section discusses key examples.

\subsection{Layout Utilities}\label{sec:utilities}
Without linear layouts, Triton's layout properties were informally defined and implemented on a case-by-case basis, leading to subtle errors and suboptimal code.
Below, we highlight two cases where linear layouts simplify this process and enhance the robustness of code generation.

\paragraph{Contiguous elements}\label{par:contig}
Computing the number of contiguous elements per thread is essential for vectorization when loading/storing tensor elements from/to global memory.
Previously, Triton heuristically identified the fastest-running dimension, assuming it determined contiguous elements.
However, when a dimension contained only one element, such as the last dimension in a tensor shape of $[128,1]$, Triton disables vectorization.

Enabling vectorization for all layouts on a case-by-case basis required extensive manual effort and was difficult to verify.
With linear layouts, this computation becomes straightforward.
It reduces to finding the largest contiguous block in the logical tensor that is mapped via the identity map onto registers by the inverse of the layout.
Given a linear layout $L$, we find the largest $u$ that has $L^{-1}_{\reg}(i)= i$, for any $i \leq u$.

\paragraph{Broadcasting}\label{sec:broadcasting}
Legacy layouts, such as blocked and MMA layouts, are defined by an initial \textbf{tile} that distributes data across registers, threads, and warps.
If the tile is smaller than the associated tensor, it is replicated to cover the entire tensor, increasing register usage per thread.
Conversely, if the tile is larger, the tensor is replicated to cover the tile, meaning threads and warps can hold duplicated data in registers.
Handling this behavior in LLVM code generation, particularly for reduction and scan operations, is complex, as determining which threads hold duplicated data in an arbitrary layout is nontrivial.
This has been a persistent source of bugs in Triton over the past few years~\cite{triton-issue-3017,triton-issue-4310,triton-issue-4362}.
Linear layouts significantly simplify this process.
Tiling operations are translated to the \emph{Product} operation (\cref{def:product}).
Once a linear layout is established, identifying threads and warps with duplicated data reduces to detecting zero columns in the layout matrix.
For example, adding a zero column in $A_{reg}$ defined in \cref{sec:ll_intro} means that registers 4-7 map to the same tensor elements as registers 0-3.

\subsection{Mixed-Precision Matrix Multiplication}\label{sec:mxfp}

Using low-precision data types in DL models is proven to maintain the same level of accuracy while improving performance~\cite{wang2018training,tsengtraining}, and it is often used in scenarios where usually one operand is of higher precision while another is of lower precision.
We now discuss how linear layouts make mixed-precision matrix multiplication robust and efficient.

\paragraph{Software Emulation}

New-generation GPUs, such as the NVIDIA B200 and AMD MI350x, provide native hardware support for matrix multiplication, such as MXFP4~\cite{OCP-MX-Spec-2023}, which is a quantized type where each 32 floating-point elements share a single 8-bit exponent (i.e., \emph{scale}).
Given the limited availability of such hardware at the time of writing, Triton needs to support software emulation on existing architectures.
For example, when performing $\code{mxfp4} \times \code{bf16}$, we upcast \code{mxfp4} to \code{bf16}.
Each set of 8 threads in a warp (\ie, each row of the \code{mma} layout) shares the same scale.
Achieving this functionality with legacy layouts would require implementing a new layout along with conversion operations across all distributed and memory layouts.
Alternatively, one could load exponents in a blocked layout and share them via warp shuffles, but at the cost of suboptimal performance.

\sloppy
Linear layouts provide a better solution.
By defining shape transformations (\ie, \code{tt.reshape}, \code{tt.transpose}, and \code{tt.broadcast}) for scale broadcasting, the layout engine automatically determines the correct layout for loading scales, while generic shared memory loads handle the rest.
This approach is also exposed at the Python API level, providing higher flexibility.

\paragraph{Data Shuffling}\label{par:data-shuffling}
Loading low-precision data and then upcasting before invoking Tensor Core instructions often results in inefficiencies.
For example, when performing $\code{mxfp4} \times \code{bf16}$, the \code{mxfp4} data cannot be loaded using vectorized instructions since the corresponding \code{wgmma} instructions require two registers per thread for each row in the operands.
To optimize performance, we can pre-shuffle the higher-precision tensor operand (\code{bf16}) in HBM before computation to enable wider vectorization for the lower-precision tensor operand (\code{mxfp4}). \footnote{Similar optimizations can be applied to \code{mma} without pre-shuffle since it accepts both operands on registers}
The Machete framework~\cite{Wilkinson2024Machete} implemented this solution using several thousand lines of code and a heavy CUTLASS~\cite{nvidia_cute} dependency.
With linear layouts, this optimization can be achieved at the language level in just five lines of Python using shape operations.


\subsection{Using SIMD Hardware Primitives}\label{sec:hw_primitives}
SIMD instructions are fundamental to modern hardware for improving data throughput.
We have discussed vectorized global memory operations and \code{mma}/\code{wgmma} operations in \cref{sec:utilities}, both of which require tensors to follow specific layouts that are constructed from small \textbf{tile}s compatible with SIMD instructions.
In this section, we discuss using efficient SIMD instructions to map one layout to another.
\begin{theorem}
    Given a layout $L$, an instruction with tile $T$ can lower it if $L \ldiv T$ exists.
\end{theorem}
\begin{proof}
    It follows from the definition of the tile $T$ and left division (cf., \cref{def:ldivision}).
\end{proof}

\paragraph{Shared Memory Load and Store}
Mapping registers from a distributed layout to the corresponding MMA swizzled layout using SIMD instructions can enable fast shared memory loads and stores.
Performing this mapping generically is challenging in the legacy Triton layout system, as it requires a unique implementation for each layout pair and only supports a subset of layouts, often leading to errors or even silent failures in complex programs.

Linear layouts offer an elegant, generic solution.
Given a memory layout represented by an invertible matrix $A$ (\cf~\cref{def:memory}) that maps offsets to the logical tensor, and a distributed layout $B$ that maps registers, threads, and warps to the same space, the required mapping reduces to computing $L = A^{-1}\circ B$.
Once $L$ is determined, we can assess whether certain SIMD instructions are compatible with the layout by constructing a corresponding tile $T$ and $L\ldiv T$ exist.

\textit{Vectorized \code{ld.shared}/\code{st.shared}}. The tile for vectorized shared memory instructions of size $2^n$ bits (typically $32$, $64$, or $128$) is given by the identity mapping from registers to memory offsets of size $n \times n$.

\textit{\code{ldmatrix}/\code{stmatrix}}. These instructions require each thread to handle $4$ contiguous bytes, with $8$ groups of $4$ threads collaborating to store a row each. For an element type of byte width $w$, the tile is given by $\id^{\reg,\offset}_k \times \id^{\thread,\offset}_2$, for $k = \log_2 \tfrac{4}{w}$ where $\id_k$ is the $k \times k$ identity matrix.


\paragraph{Generalized Vectorization}
If the layout $L$ does not have the structure to be divided by $T$, we can adjust it by permuting the registers.
For example, if the layout is column-major, vectorization would not be directly possible.
Instead, we define $L' = P_{\reg} L$, where $P_{\reg}$ permutes the registers.
Since the division algorithm processes the columns of $L$ and $T$ sequentially, we can determine $P_{\reg}$ while computing the division.

\subsection{Optimal Codegen for Layout Conversions}\label{sec:conversion}
Given distributed layouts $A$ and $B$, we can convert the tensor/hardware resource mapping from $A$ to $B$. Treating $A$ and $B$ as representing vectors in $\F^d$ (flattening the logical tensor $\F^{d_1}\times \dots \times \F^{d_r} \iso \F^d$), we define the sets $L_\reg, L_\thread, L_\warp$ as the columns of a distributed layout $L$ that act on registers, threads, and warps. By~\cref{thm:distr}, these elements are distinct powers of two or zeros.

The conversion is given by $B^{-1}\circ A$.
While $B$ need not be invertible, it is surjective as it represents the entire logical tensor, so a right inverse exists.
We select $B^{-1}\circ A$ to satisfy:
\begin{enumerate}
    \item \textbf{Minimizing inter-warp or inter-thread data movement:} If $A_i = B_i$, then $\pa{B^{-1}\circ A}_i$ is the identity for $i \in \set{\reg, \thread, \warp}$.
    \item \textbf{Promoting broadcasting:} 
    The linear system $BX = A$ can have multiple solutions, such as when $B$ is a distributed layout where the same tensor element is broadcast across registers. 
    To pick a unique one, we set the slack variables in the linear system to zeros to produce a solution $X$ whose Hamming weight~\cite{wiki-hamming-weight}---the number of 1-bits in $X$---is minimal.
    Intuitively, we make all the elements pointing to the same value in the logical tensor read from the same input execution unit.
\end{enumerate}

\paragraph{Intra-thread Data Exchange}
$\pa{B^{-1}\circ A}_\reg$ is the register permutation needed to transform $A$ into $B$.

\begin{figure}[t]
    \centering
    \includegraphics[width=1.0\linewidth]{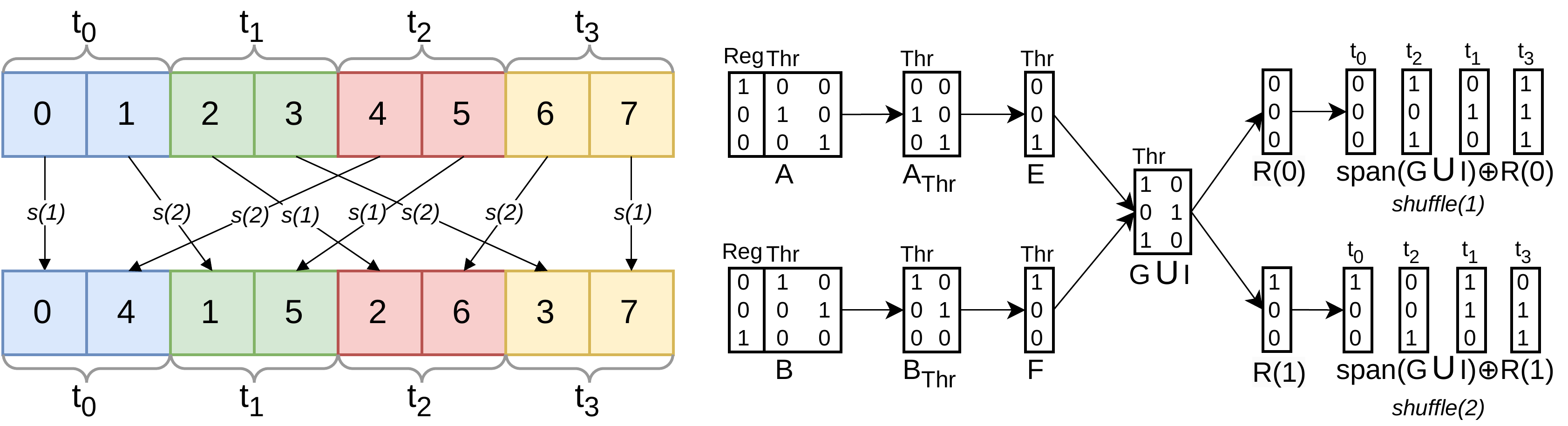}
    \caption{A step-by-step illustration of layout conversion through warp shuffles. $t_{i}$ denotes thread $i$. $s(1)$ and $s(2)$ denote the first and the second shuffle round, respectively. For simplicity, we illustrate with four threads only without loss of generality to any warps containing the power of 2 threads.}
    \label{fig:convert-warp-shuffle}
\end{figure}

\paragraph{Intra-warp Data Exchange}
If $(B^{-1}\circ A)_\warp$ is the identity, data exchange can be performed using warp shuffles.
For simplicity, assume there is no broadcasting in $A$ or $B$.
We divide the process into two steps:

\textbf{1. Determining the vectorization size}
   The number of bytes that can be exchanged per warp shuffle depends on the vectorization of $\pa{B^{-1}\circ A}_\reg$.
   Specifically, if $n = \card{A_\reg \cap B_\reg}$, then each warp shuffle can transfer up to $2^n$ elements. Let $V \subseteq A_\reg \cap B_\reg$ be the largest subset that can be exchanged in a single warp instruction---typically $32$ bits on NVIDIA and AMD hardware.

\textbf{2. Tiling and exchanging elements}
Since we are exchanging elements defined by the basis vectors of $V$, we must tile the complement of the subspace $\sspan\pa{V}$. Each shuffle operation enables a thread to send and receive $2^{\card{V}}$ elements. To determine which elements should be exchanged, define
   \[
   I = A_\thread \cap B_\thread \qquad
   E = A_\thread \backslash I \qquad
   F = B_\thread \backslash I
    \]
   $I$ contains the vectors in both $A_\thread$ and $B_\thread$, which do not have to perform data exchange.
   Then, we can take these vectors out of $A_\thread$ and $B_\thread$ to obtain $E$ and $F$.
   Since there is no broadcasting, we have that $\card{E} = \card{F}$. After choosing an ordering (e.g., ascending order) for $E$ and $F$, we define $G$ as
   \[
    G = \set{e_i \oplus f_i \mid e_i \in E, f_i \in F, 1 \leq i \leq \card{E}}.
   \]
   $G$ is a basis of the subspace such that each element of this subspace belongs to a different thread of $A$ and a different thread of $B$.
   $V \cup I \cup G$ forms a basis of the subspace containing elements that will participate in the first shuffle round.

   Since we have to tile the complement of subspace $\sspan(V)$, we extend the basis $V \cup I \cup G$ to a basis of the whole space $\F^d$.
   We call this extension $R$, and we see $R$ as a mapping from $0 \dots 2^{\card{R}} - 1$ to $\F^d$. Then, for each $i$, the affine space $R(i) \oplus \sspan\pa{V \cup I \cup G}$ contains exactly one vectorized element per thread in layouts $A$ and $B$, so we can exchange the elements in $2^{\card{R}}$ rounds, shuffling the elements in each round.

    \Cref{fig:convert-warp-shuffle} demonstrates an example that uses warp shuffles. 
    $V$ is empty in this case.
    Next, to complete the space $\F^{3}$, we define $R(0) = [0, 0, 0]^{T}$ and $R(1) = [0, 1, 0]^{T}$, and get $\sspan(G)$.
    Because $\sspan$ of a set of vectors is the set of all linear combinations of the vectors in this set, and $V$ and $I$ are empty, we have $\sspan(V \cup I \cup G) = \{[0\,0\,0]^{T}, [1\,0\,1]^{T}, [0\,1\,0]^{T}, [1\,1\,1]^{T}\}$.
    The result of $R(i) \oplus \sspan\pa{G}$ represents the location of elements that will be involved in shuffle round $i$.
    In each round, every thread sends and receives only one element.

\paragraph{Optimal Swizzling}

\begin{figure*}
    \centering
    \includegraphics[width=0.8\linewidth]{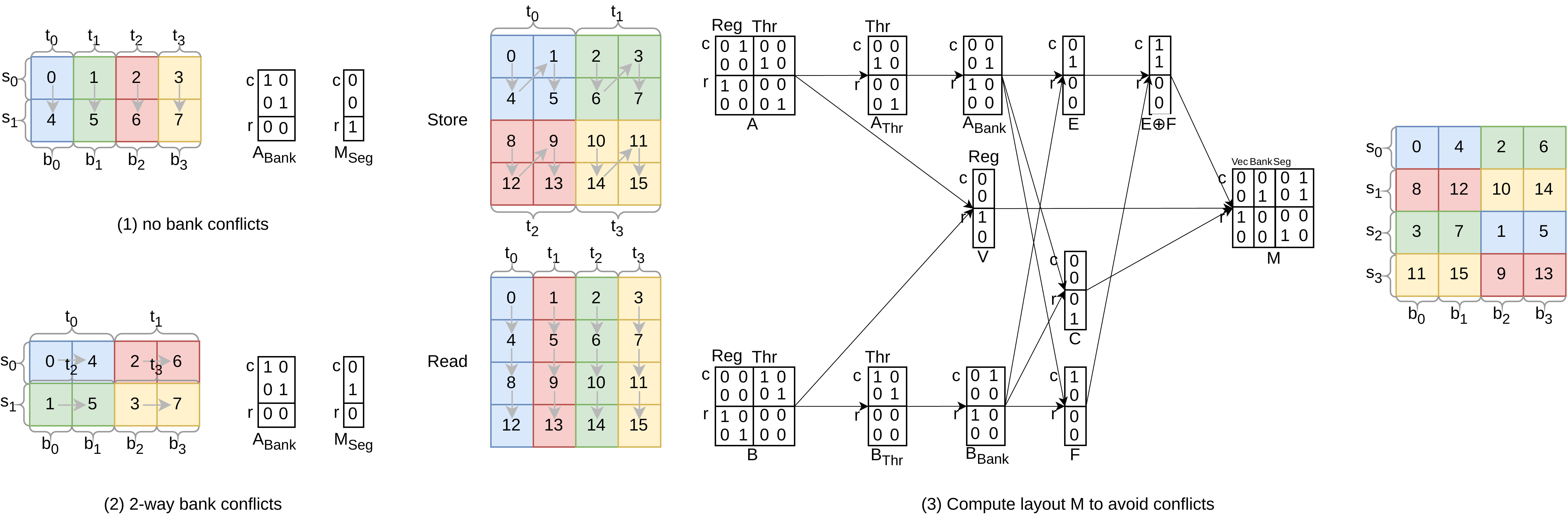}
    \caption{Bank conflicts and swizzling. $t_{i}$ denotes thread $i$, $s_{j}$ denotes segment $j$, and $b_{k}$ denotes bank $k$. $c$ and $r$ denotes row and column index. 
    For simplicity, we illustrate with four threads, four segments, and each segment contains four banks.}
    \label{fig:convert-swizzle}
\end{figure*}

We now present an algorithm that computes an optimal swizzled layout that maximizes read/write vectorization while minimizing bank conflicts for arbitrary linear layouts. 

We represent the shared memory layout as a map
\[
    \deffun{M : \F^v \times \F^b \times \F^s -> \F^d;},
\]
where $s = d - v - b$. The first space $\vect$ corresponds to vectorization, the second space $\bank$ represents memory banks with each segment, and the third space $\seg$ corresponds to segment index.

To minimize bank conflicts, each bank belonging to distinct segments should be accessed by distinct threads.
Let $P = \sspan(M_{\vect} \cup A_{\thread})$, we aim to identify the largest subspace $H$ such that $P \cap \sspan(H) = \{0\}$.
If $P$ overlaps with $\sspan(H)$, it implies that at least two threads may access the same bank on different segments.
In the worst case, if the segment space $\sspan(H)$ equals $P$, all threads will access the same bank across different segments.
Figure~\ref{fig:convert-swizzle} (2) has 2-way conflicts because $t_{0}$ and $t_{2}$ access the same bank as well as $t_{1}$ and $t_{3}$.

In the following, we describe the bank conflict optimization algorithm when two layouts, $A$ and $B$.
We first define the vectorization set $V$ of size $2^v$ by choosing a basis of $A_\reg \cap B_\reg$ as done for warp shuffles.
For a data type with byte width $w$, let $b$ be the logarithm of the number of vectorized elements needed to cover all the shared memory banks.
On modern GPUs, this is $b = \log_2 \tfrac{128}{2^vw}$.
On NVIDIA GPUs, if vectorization modifiers (e.g., \texttt{.v4}) are used, transactions involving more than 128 bytes will be split into multiple 128-byte transactions, so we generate two new layouts $A_{\bank}$ and $B_{\bank}$ by taking out the last few $\log_2 \max(1, \frac{2^{v}w}{4})$ vectors from $A_{\thread}$ and $B_{\thread}$ respectively.
Next, we define 

Define
\[
    P = \sspan\pa{S_{\vect} \cup A_{\bank}} \cup \sspan\pa{S_{\vect} \cup B_{\bank}}.
\]
To minimize bank conflicts, we are interested in finding the largest subspace $H$ such that $P \cap \sspan\pa{H} = \set{0}$.
We define
\[
E = A_\bank \backslash B_\bank, \quad F = B_\bank \backslash A_\bank.
\]
Without loss of generality, assume that $\card{E} \leq \card{F}$.
We then enumerate their elements following a chosen order and construct
\[
    H = \set{e_i \oplus f_i | e_i \in E, f_i \in F, 1 \leq i \leq \card{E}}.
\]
Next, we construct a basis $C$ as a complementing subspace of $P$ and determine the columns of $M_{\seg}$ as follows:
\begin{itemize}
\item If $\card{H} + \card{C} \geq s$, we select $s$ vectors from $H \cup C$.
\item If $\card{H} + \card{C} < s$, bank conflicts are unavoidable. We add the remaining $s - \card{H} - \card{C}$ vectors from $A_\bank$.
\end{itemize}

Finally, we choose $S_\bank$ by completing the columns of $M$ into a basis of $\F^d$ similar to the warp shuffling process. 
$M$ is the swizzled layout that minimizes read and write bank conflicts provided maximal vectorization.
We demonstrate the workflow of this algorithm in Figure~\ref{fig:convert-swizzle}.
Reads and writes are split into four transactions without bank conflicts.
For example, for memory reads in Figure~\ref{fig:convert-swizzle} (3), in the first transaction, $t_{0}$ reads 0 ($b_{0}$) and 4 ($b_{1}$), and $t_{1}$ reads 1 ($b_{2}$) and 5 ($b_{3}$). In the second transaction, $t_{2}$ reads 2 ($b_{2}$) and 6 ($b_{3}$), and $t_{3}$ reads 3 ($b_{0}$) and 7 ($b_{1}$).

\subsection{Optimized Codegen for Gather}\label{subsec:gather}
The \code{tl.gather} operator extracts specific elements from a source tensor (\code{src}) along a given axis (\code{axis}) using indices from the \code{index} tensor.
If all elements along the \code{axis} dimension of \code{src} and \code{index} reside within the same warp, we can optimize the operation using warp shuffles.
This is determined by checking whether all elements of $L_{\warp}^{\code{axis}}$ are zero, where $L$ is the layout of both \code{src} and \code{index}.

To exchange elements between threads, for each position, $pos$, along \code{axis}, we first read $\texttt{index}(pos)$ to obtain the location of the source and use $L(\texttt{index}(pos))_{\reg}$ and $L(\texttt{index}(pos))_{\thread}$ to identify the register and thread that holds the source value.
Then, we perform $n$ rounds of warp shuffles, $n = 2^{\card{L_{\thread}^{\code{axis}}}}$.
In each round, a thread sends its $i$-th value and receives a value from the source thread $L(\texttt{index}(pos))_{\thread}$.
The received value is stored only if $i = L(\texttt{index}(i))_{\reg}$.

\section{Evaluation}\label{sec:evaluation}
We compared our optimized version of Triton, which integrates linear layout-based optimizations (Triton-Linear), with the baseline Triton that does not incorporate these optimizations.  
The key differences between Triton and Triton-Linear are as follows:
\begin{itemize}
    \item Triton uses legacy data layouts, which do not support utilities for arbitrary distributed layouts or  conversions between them, making it prone to bugs.
    \item Triton does not incorporate optimized code generation as described in \cref{sec:codegen}.
    For example, layout conversions always go through shared memory, with limited use of efficient hardware primitives.
\end{itemize}

In the following, we first compare the test pass rate and performance between Triton and Triton-Linear using synthetic micro-benchmarks.
The running time is obtained by repeating each benchmark 10 times and reporting the median value.
Next, we compare the performance of the two versions using individual kernels in TritonBench~\cite{tritonbench}, with the running time reported by TritonBench's reporting system.  
We evaluated the performance on three distinct platforms, as detailed in \cref{tab:platforms}.

\begin{table}[t]
\centering
\scriptsize
\caption{Hardware Platforms Evaluated}
\label{tab:platforms}
\begin{tabular}{lccc}
\toprule
\textbf{Platform} & \textbf{GPU Model} & \textbf{Memory} & \textbf{Notes} \\
\midrule
RTX4090   & NVIDIA RTX4090 & 24GB GDDR6X & Consumer GPU \\
GH200  & NVIDIA GH200  & 80GB HBM2e  & Data center GPU \\
MI250 & AMD MI250   & 64GB HBM2   & Data center GPU \\
\bottomrule
\end{tabular}
\end{table}

\subsection{Micro-Benchmarks}

\paragraph{Hyperparameters} All microbenchmarks, except for Mixed Precision Matmul, are executed using four warps and a single CTA.
The Mixed Precision Matmul benchmark uses four warps per CTA, with the number of CTAs varying based on the input size.

\paragraph{Load/Store Contiguity}
We synthesized a benchmark that loads and stores tensors of varying sizes in the last dimension with different data types.
The pass rates of Triton and Triton-Linear are shown in Table~\ref{tab:micro-load-store}.
We observe that Triton, using legacy layouts, fails to identify the maximum number of contiguous elements when they span multiple dimensions, even though each thread can access these elements contiguously. 
In contrast, linear layouts enable identifying the maximum number of contiguous elements across dimensions, resulting in up to a $7\times$ increase in the bitwidth accessed by load/store instructions.

\begin{table}[tp]
\tiny
\centering
\renewcommand{\arraystretch}{1.2} 
\caption{Comparison of load/store instructions and bitwidths across different shapes and data types.}
\begin{tabular}{|c|cc|cc|}
\hline
\rowcolor[HTML]{C0C0C0} 
& \multicolumn{2}{c|}{\textbf{Load/Store Inst}} & \multicolumn{2}{c|}{\textbf{Bitwidth}} \\ \hline
\rowcolor[HTML]{C0C0C0} 
\textbf{Tensor Type} & \textbf{Triton} & \textbf{Triton-Linear} & \textbf{Triton} & \textbf{Triton-Linear} \\ \hline
$[512,1] \times f8$  & \multicolumn{1}{c|}{v1.b32}  & v1.b32  & \multicolumn{1}{c|}{32}   & 32   \\ \hline
$[512,2] \times f8$  & \multicolumn{1}{c|}{v1.b16}  & v4.b32  & \multicolumn{1}{c|}{16}   & \cellcolor[HTML]{9AFF99}128 ($\uparrow$ 700\%)  \\ \hline
$[512,4] \times f8$  & \multicolumn{1}{c|}{v1.b32}  & v4.b32  & \multicolumn{1}{c|}{32}   & \cellcolor[HTML]{9AFF99}128 ($\uparrow$ 400\%)  \\ \hline
$[512,8] \times f8$  & \multicolumn{1}{c|}{v2.b32}  & v4.b32  & \multicolumn{1}{c|}{64}  & \cellcolor[HTML]{9AFF99}128 ($\uparrow$ 100\%)  \\ \hline
$[512,16] \times f8$ & \multicolumn{1}{c|}{v4.b32}  & v4.b32  & \multicolumn{1}{c|}{128}  & 128  \\ \hline
$[512,1] \times f16$  & \multicolumn{1}{c|}{v2.b32}  & v2.b32  & \multicolumn{1}{c|}{64}   & 64   \\ \hline
$[512,2] \times f16$  & \multicolumn{1}{c|}{v1.b32}  & v4.b32  & \multicolumn{1}{c|}{32}   & \cellcolor[HTML]{9AFF99}128 ($\uparrow$ 300\%)  \\ \hline
$[512,4] \times f16$  & \multicolumn{1}{c|}{v2.b32}  & v4.b32  & \multicolumn{1}{c|}{64}   & \cellcolor[HTML]{9AFF99}128 ($\uparrow$ 100\%)  \\ \hline
$[512,8] \times f16$  & \multicolumn{1}{c|}{v4.b32}  & v4.b32  & \multicolumn{1}{c|}{128}  & 128  \\ \hline
$[512,16] \times f16$ & \multicolumn{1}{c|}{v4.b32}  & v4.b32  & \multicolumn{1}{c|}{128}  & 128  \\ \hline
\end{tabular}
\label{tab:micro-load-store}
\end{table}

\paragraph{Broadcasting}

\begin{table}[tp]
\tiny
\centering
\renewcommand{\arraystretch}{1.2} 
\caption{Comparison of layout support and the number of shared memory instructions.}
\begin{tabular}{|c|cc|cc|}
\hline
\rowcolor[HTML]{C0C0C0} 
& \multicolumn{2}{c|}{\textbf{Pass Rate}} & \multicolumn{2}{c|}{\textbf{\#Shared Memory Insts}} \\ \hline
\rowcolor[HTML]{C0C0C0} 
\textbf{Layout} & \textbf{Triton} & \textbf{Triton-Linear} & \textbf{Triton} & \textbf{Triton-Linear} \\ \hline
Blocked    & \multicolumn{1}{c|}{20/20} & 20/20  & \multicolumn{1}{c|}{5888}  & \cellcolor[HTML]{9AFF99}1388 ($\downarrow$ 76\%) \\ \hline
MMA        & \multicolumn{1}{c|}{20/20} & 20/20  & \multicolumn{1}{c|}{5914}  & \cellcolor[HTML]{9AFF99}3517 ($\downarrow$ 40\%) \\ \hline
MMA Input        & \multicolumn{1}{c|}{0/10}  & \cellcolor[HTML]{9AFF99}10/10 & \multicolumn{1}{c|}{N/A}   & \cellcolor[HTML]{9AFF99}5884 \\ \hline
Sliced\textless{}Blocked\textgreater{} & \multicolumn{1}{c|}{20/20} & 20/20  & \multicolumn{1}{c|}{6703}  & \cellcolor[HTML]{9AFF99}4687 ($\downarrow$ 30\%) \\ \hline
Sliced\textless{}MMA\textgreater{}     & \multicolumn{1}{c|}{0/10}  & \cellcolor[HTML]{9AFF99}10/10 & \multicolumn{1}{c|}{N/A}   & \cellcolor[HTML]{9AFF99}320 \\ \hline
Sliced\textless{}MMA Input\textgreater{}     & \multicolumn{1}{c|}{0/10}  & \cellcolor[HTML]{9AFF99}10/10 & \multicolumn{1}{c|}{N/A}   & \cellcolor[HTML]{9AFF99}545 \\ \hline
Custom     & \multicolumn{1}{c|}{0/10}  & \cellcolor[HTML]{9AFF99}10/10 & \multicolumn{1}{c|}{N/A}   & \cellcolor[HTML]{9AFF99}913 \\ \hline
\end{tabular}
\label{tab:micro-broadcasting}
\end{table}

As discussed in \cref{sec:broadcasting}, using linear layouts, we can correctly identify threads and warps with duplicated data, helping to avoid redundant load and store instructions.  
We designed a micro-benchmark to enumerate the most common layouts in Triton and applied a reduction operation across tensors with the following shapes: \([128, 16]\), \([128, 128]\), \([32, 128]\), \([32, 32]\), and \([16, 16]\).
Experiment results in \cref{tab:micro-broadcasting} demonstrate that Triton-Linear not only supports reduction operations across all layout combinations but also reduces the number of shared memory store instructions by up to 76\%.

\paragraph{Mixed Precision Matmul}

\begin{table}[tp]
\tiny
\centering
\renewcommand{\arraystretch}{1.2} 
\caption{Pass rate comparison for different data type pairs.}
\begin{tabular}{|c|c|c||c|c|c|}
\hline
\rowcolor[HTML]{C0C0C0} 
 & \multicolumn{2}{c||}{\textbf{Pass Rate}} & & \multicolumn{2}{c|}{\textbf{Pass Rate}} \\ \hline
\rowcolor[HTML]{C0C0C0} 
 \textbf{Data Type} & \textbf{Triton} & \textbf{Triton-Linear} & \textbf{Data Type}  & \textbf{Triton} & \textbf{Triton-Linear} \\ \hline
i16/f16  & 32/64  &\cellcolor[HTML]{9AFF99}64/64  & i16/f32  & 32/32  & 32/32  \\ \hline
i16/f64  & 32/32  & 32/32  & i16/f8   & 36/96  & 96/96\cellcolor[HTML]{9AFF99}  \\ \hline
i32/f16  & 32/32  & 32/32  & i32/f64  & 16/32  & 32/32\cellcolor[HTML]{9AFF99}  \\ \hline
i32/f8   & 18/48  &\cellcolor[HTML]{9AFF99} 48/48  & i64/f16  & 32/32  & 32/32  \\ \hline
i64/f32  & 16/32  &\cellcolor[HTML]{9AFF99} 32/32  & i64/f8   & 18/48  &\cellcolor[HTML]{9AFF99} 48/48  \\ \hline
i8/f16   & 36/96  &\cellcolor[HTML]{9AFF99} 96/96  & i8/f32   & 18/48  &\cellcolor[HTML]{9AFF99} 48/48  \\ \hline
i8/f64   & 18/48  &\cellcolor[HTML]{9AFF99} 48/48  & i8/f8    & 30/144 &\cellcolor[HTML]{9AFF99} 144/144 \\ \hline
\end{tabular}
\label{tab:micro-mixed-precision}
\end{table}

We built two micro-benchmarks to compare Triton-Linear with Triton for mixed-precision matrix multiplications.
First, we enumerated all common tensor data types used in Triton in pairs, testing the correctness of a simple matrix multiplication kernel across different shapes.  
As shown in \cref{tab:micro-mixed-precision}, we observe that Triton fails in many cases, achieving an overall pass rate of only 46.6\% out of the total 784 cases, whereas Triton-Linear successfully passes all test cases.  
The main reason behind this is that Triton does not correctly implement matrix multiplication for small shapes and low-precision data types.  
In fact, Triton does not support any MMA layouts with more than 32-bit consecutive elements in the last dimension of the tile.
In contrast, linear layouts provide a solid foundation for code generation, ensuring support for all valid distributed layouts in matrix multiplication.

The second micro-benchmark we constructed evaluates the performance gains achieved using the data shuffling optimization described in \cref{par:data-shuffling}.  
We fixed one operand as \code{mxfp4} while varying the precision of the other operand. 
As shown in \cref{fig:micro-mxfp4}, Triton-Linear consistently outperforms Triton across different tensor shapes and data types due to the higher throughput enabled by vectorized shared memory instructions.  
Notably, the $\text{mxfp4} \times \text{f16}$ series of experiments shows a higher speedup (1.87$\times$), as we also addressed an issue where Triton did not utilize \code{wgmma} for \code{f16} in mixed-precision cases.

\begin{figure}[tp]
    \centering
    \includegraphics[width=0.9\linewidth]{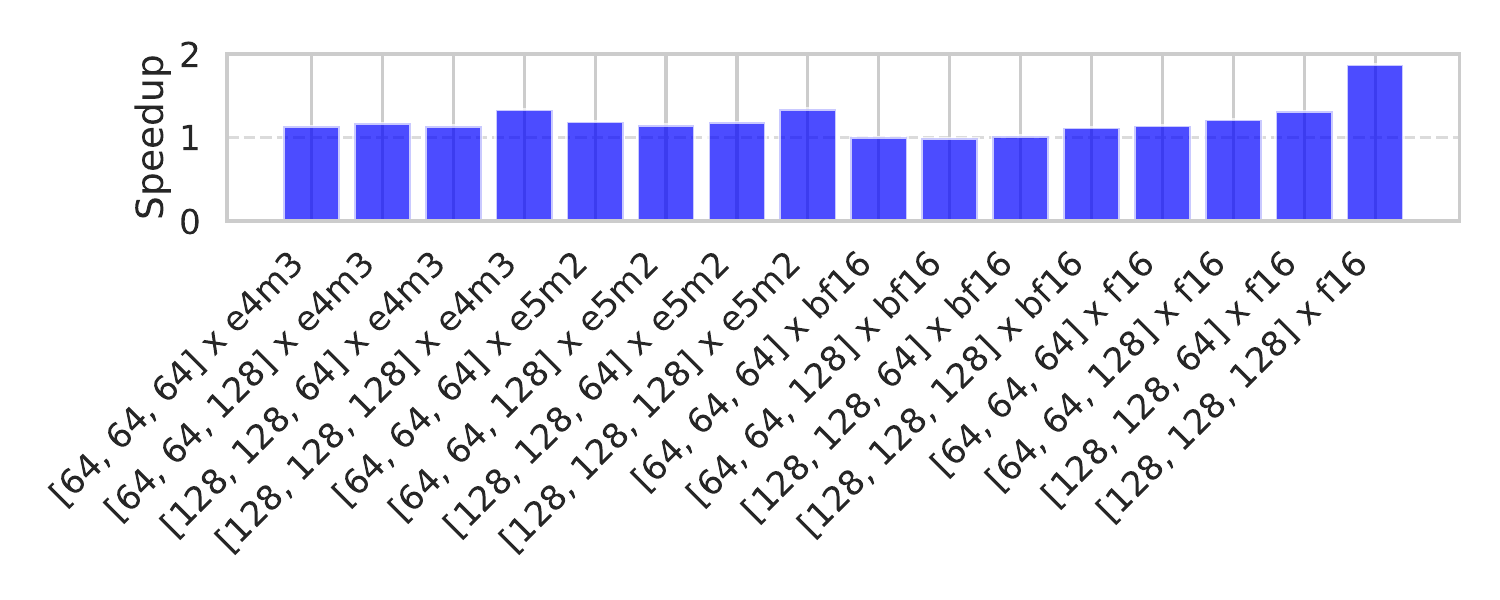}
    \caption{Speedups of MXFP4 matrix multiplications across different shapes and data types on GH200.}
    \label{fig:micro-mxfp4}
\end{figure}

\begin{figure}
    \centering
    \includegraphics[width=0.9\linewidth]{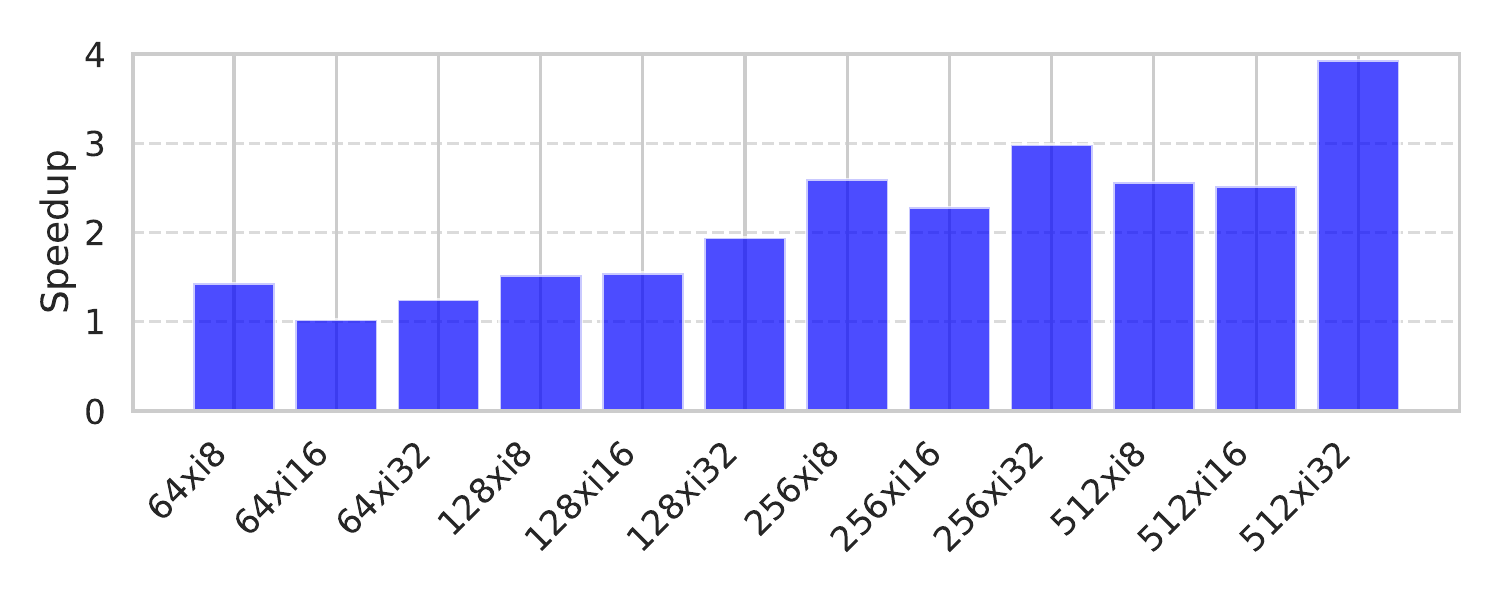}
    \caption{Speedups of layout conversions across different shapes and data types on GH200. }
    \label{fig:convert-layout}
\end{figure}

\begin{figure}[tp]
    \centering
    \includegraphics[width=0.9\linewidth]{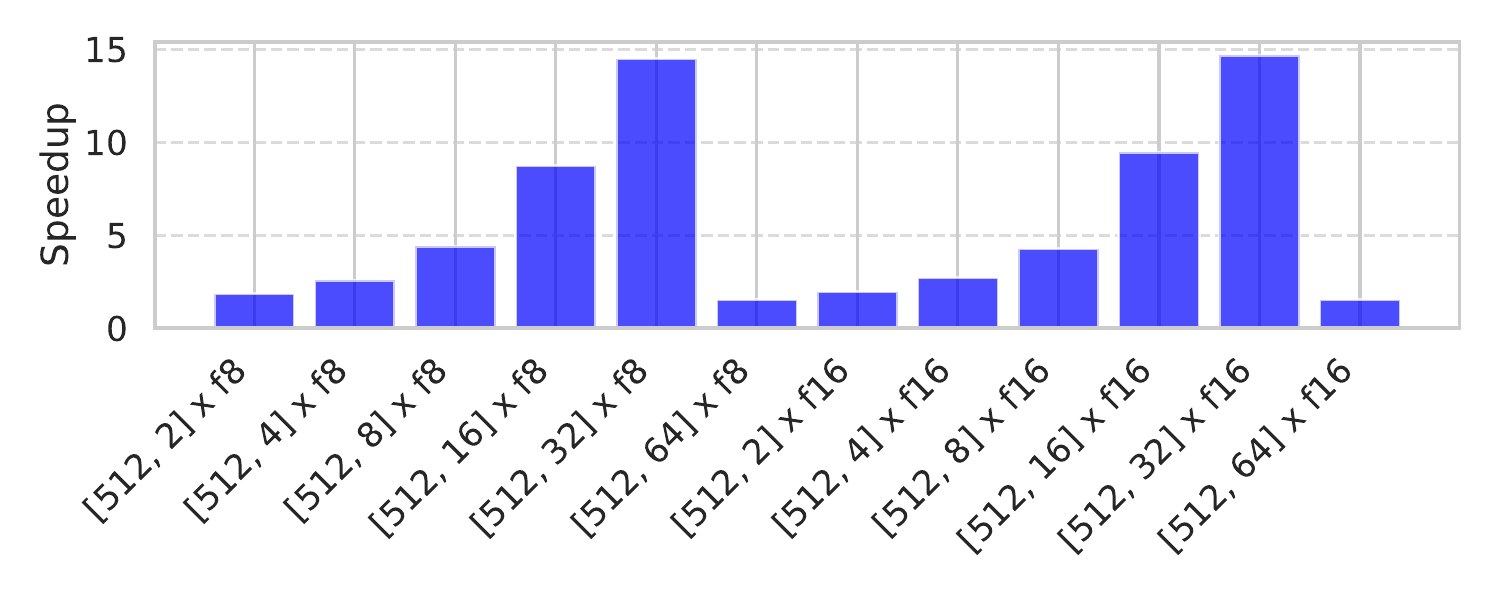}
    \caption{Speedups of the gather operator across different shapes and data types on GH200.}
    \label{fig:micro-gather}
\end{figure}

\paragraph{Layout Conversion} We compared the performance of Triton and Triton-Linear when warp shuffles are used for layout conversions.  
Our benchmark evaluated tensors of varying sizes and data types.  
As shown in \cref{fig:convert-layout}, Triton-Linear consistently outperforms Triton, which always uses shared memory-based layout conversion, achieving speedups of up to 3.93$\times$.

\paragraph{Gather} We evaluated the performance improvement of the \code{gather} operator when warp shuffles are used, comparing it to Triton's implementation, which always uses shared memory. 
\cref{fig:micro-gather} shows that Triton-Linear achieves a maximum speedup of 14.20$\times$ over Triton.  
Interestingly, as the gathered dimension increases, the speedup drops after a certain point (e.g., $[512, 32]$), because the overhead of emitting multiple rounds of warp shuffles outweighs the benefits of eliminating shared memory accesses.

\subsection{Real Benchmarks}

\begin{figure*}[tp]
    \centering
    \includegraphics[width=1.0\linewidth]{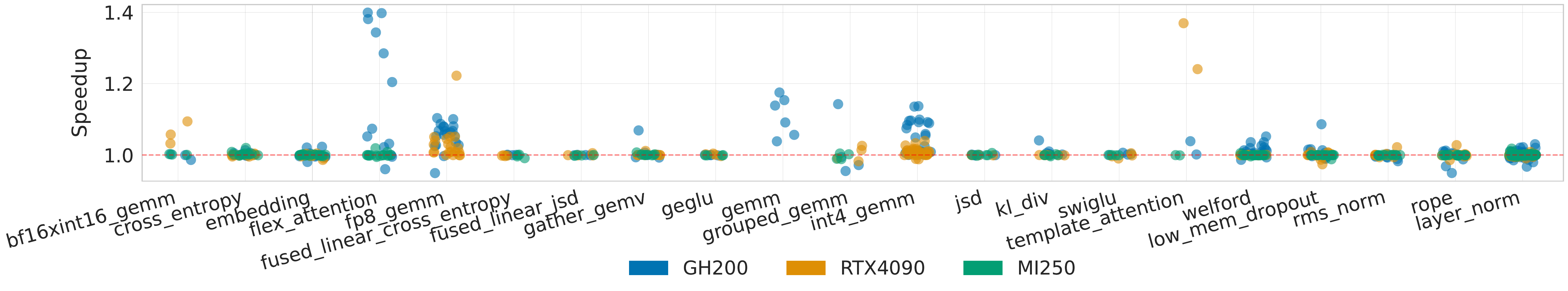}
    \caption{Speedups of real benchmarks on RTX4090, GH200, and MI250.}
    \label{fig:real-bench}
\end{figure*}
We ran 21 benchmarks in TritonBench on three different platforms to compare the performance of Triton with that of Triton-Linear.
We show the performance gain of Triton-Linear on three platforms in \cref{fig:real-bench}.
Because each benchmark has multiple inputs, \textbf{totaling 265 cases}, we use circles to indicate the speedup of each case. 
Note that benchmarks are not all available on each platform due to hardware limitations.
For example, some benchmarks require large shared memory available only on GH200, while several kernels use tensor descriptors that rely on TMA engines~\cite{nvidia_ptx_isa}, which are absent on both RTX4090 and MI250.
In addition, speedups lower than 1.0 are mostly caused by runtime noise in benchmarks when small inputs are used.

On GH200, we achieved speedups ranging from 0.96$\times$ to 1.40$\times$.
The benchmarks with the most significant speedups are \emph{int4\_gemm}, \emph{gemm}, and \emph{flex\_attention}. 
We observe that efficient hardware primitives, such as \code{ldmatrix} and \code{stmatrix}, are widely utilized in layout conversion and shared memory load and store operations within these kernels.
For \emph{welfrod}, Triton-Linear is able to detect the conversion between ``equivalent'' layouts, allowing the conversion to be lowered to a no-op.  
These optimizations are not possible in the legacy layout system, as it cannot directly compare layouts of different kinds (e.g., Blocked and Sliced layouts).
We plot \autoref{tab:operation-analysis} to show the distribution of \emph{convert\_layout}, \emph{local\_load}, and \emph{local\_store} operations in Triton's GPU IR and confirm that the benefits of Linear Layouts come from optimizing the cost associated with these operations.

\begin{table}[ht]
    \centering
    \scriptsize
    \caption{Distribution of local (shared) memory and convert layout operations in each benchmark. Benchmarks with no relevant operations are omitted.}
    \begin{tabular}{lrrrr}
        \hline
        Operation & \#Load & \#Store & \#Convert \\
        \hline
        gemm & 76 & 18 & 54 \\
        bf16xint16\_gemm & 22 & 14 & 32 \\
        int4gemm & 9 & 3 & 6 \\
        template\_attention & 2 & 4 & 2 \\
        fp8\_gemm & 4 & 0 & 16 \\
        welford & 0 & 0 & 8 \\
        gather\_gemv & 0 & 0 & 8 \\
        grouped\_gemm & 0 & 0 & 4 \\
        rope & 0 & 0 & 2 \\
        embedding & 0 & 0 & 1 \\
        \hline
    \end{tabular}
    \label{tab:operation-analysis}
\end{table}

On RTX4090, we achieved speedups from 0.97$\times$ to 1.37$\times$.
We achieved a higher speedup on \code{template\_attention} due to the difference between \code{mma} (RTX4090) and \code{wgmma} (GH200) instructions.
In this case, a \code{tt.dot} operation has the left operand defined outside of the loop, repeatedly loading data from the same address, thus both \code{ldmatrix} and regular shared memory instructions can achieve high throughput.
While the right operand is updated in each iteration, \code{wgmma} accesses it directly in the shared memory, only on RTX4090 it will be lowered into \code{ldmatrix} after our optimizations. 
As a result, the achieved speedup on GH200 is comparatively lower.
On MI250, we achieved a speedup from 1.00$\times$ to 1.03$\times$.
In general, Triton-Linear achieves lower speedups on AMD GPUs than NVIDIA GPUs for the lack of efficient hardware primitives such as \code{ldmatrix}.



\section{Related Work}

\paragraph{DL Compilers}
Many DL compilers~\cite{tvm_relax,Glow2018,XLA2020,TASO2019,Ansor2020,BladeDISC2023} focus on end-to-end optimizations, including operator fusion, graph transformations, and tiling-based lowering, for improved speed and memory efficiency.  
While these compilers simplify development, determining optimal optimization policies for the entire computation graph remains challenging.  
Recently, finer-grained programming models and compilers~\cite{tvm_relax,tillet2019triton,paszke2024mosaic,thunderkittens2024} have emerged, enabling users to customize deep learning operators at the tile level.
Kernels generated by these compilers often achieve higher performance compared to those produced by end-to-end compilers, due to their greater flexibility and specialized optimizations.

\paragraph{Hardware Resource Mapping}
A large body of work~\cite{ragan2013halide, zheng2020flextensor, lee2020protuner, chen2021alt, yan2021rammer, Ansor2020, hidet202X, graphene202X, amos202X} studied the layout mapping between hardware resources and logical tensors.
However, these studies have not examined the efficiency of layout conversions and lack sophisticated code generation techniques, as well as a solid theoretical foundation.
As a result, key aspects such as mixed precision, advanced hardware primitives, swizzled layouts, and efficient layout conversion remain largely unaddressed by these approaches.
The most relevant work to ours is CuTe~\cite{nvidia_cute}.
While both CuTe and linear layouts aim to address the challenge of flexible task mapping on emerging architectures, they differ in several key aspects.
First and foremost, CuTe is primarily designed for users to manually describe layouts, whereas linear layouts are integrated into a compiler.
Second, the linear algebra framework of linear layouts enables compilers to generate efficient code for layout conversion and code lowering for many common operators, which is absent in CuTe.
Third, swizzling is inherently defined within linear layouts, whereas in CuTe, it is treated as a separate step.
Additionally, dimensions in linear layouts are labeled, whereas CuTe uses unlabeled layouts.

\paragraph{Polyhedral compilation.}
Classic polyhedral compilers such as PluTo and Polly model the mapping from loop iterators to array indices as an affine function over $\mathbb{Z}$, using integer-linear programming to satisfy dependence, liveness, and boundary constraints~\cite{Bondhugula08,Grosser12,Sarkar19,Iannetta22}.
By contrast, linear layouts employed in tile-based programming frameworks map logical-tensor coordinates to physical hardware resources using a linear function over $\mathbb{F}_2$.
Bridging these two ideas opens a path to automatically lift sequential code into accelerator kernels (e.g., Triton).

\paragraph{Triton and Related Optimizations}
Recent work has explored enhancing the performance of DL models by either leveraging Triton as a programming language or improving Triton's compiler backend.
Li et al.~\cite{li2025tritonbench} investigated the automatic construction of Triton kernels using language models.
Ansel et al.~\cite{ansel2024pytorch} converted PyTorch code to Triton through tracing and heuristic-based optimizations, and He et al.~\cite{he2025cuasmrl} improved the performance of Triton-generated code using reinforcement learning.
We believe that linear layouts can enhance these frameworks by providing a well-defined mapping between hardware resources and logical tensors.
\vspace{-0.5em}
\section{Conclusions}
Linear layouts form the first theoretical foundation and implementation for resource mapping between complex hardware components and logical tensors.  
Through our framework, we prove the completeness of linear layouts under Triton's shape operators.
We also describe efficient code generation techniques using linear layouts.
Our experiments demonstrate that linear layouts not only enhance the robustness of the Triton compiler but also deliver non-trivial performance improvements. 
The primary limitation of linear layouts is the restriction to power-of-two shapes; however, this can be mitigated by defining larger tensors and masking out-of-boundary elements.
Operations such as flipping and slicing are not expressible as linear layouts $y = Ax$, but can be captured by the simple extension of `affine layouts' $y = Ax \oplus b$.
In the future, we plan to integrate linear layouts with hardware measurements to develop a holistic performance model for autotuning kernel performance.

\begin{acks}
We thank all reviewers for their valuable feedback.  
We also thank Lei Zhang and Vinod Grover for their suggestions on improving the paper.  
This work used AMD GPUs provided by AMD.  
Keren Zhou's research was supported in part by NSF Award 2411134.
\end{acks}

\bibliographystyle{ACM-Reference-Format}
\balance
\bibliography{reference}

\onecolumn

\section{Appendix}\label{sec:appendix}
Here we present the proofs of the results we discussed in the main text.

\subsection{Layout Engine}
\textbf{Notation}.
As we will be working with labelled input and output dimensions, we will denote by $\id_k^{i, j}$ the \textbf{identity map} of shape $k\times k$ going from input dimension $i$ (\eg, $\reg, \thread, \warp$) to the $j$-th output dimension (often the logical tensor). More formally, since all these spaces have a canonical basis, it maps identically the subspace generated by the first $k$ bases from the input space into the subspace generated by the first $k$ basis of the output space.

We start with the proof that blocked layouts are linear layouts. This is one of those proofs that are trivial, but its simplicity gets hidden behind all the objects that are needed to formalize it.
\begin{proposition}
    Blocked layouts are linear layouts.
\end{proposition}
\begin{proof}
    For a blocked layout associated to a tensor of shape $(d_1, \dots, d_\ell)$, consider the tuples of length $\ell$ $R, T, W$ representing the $\log_2$ of the number of registers, threads, and warps per dimension. Note that $R_i + T_i + W_i = d_i$. A blocked layout also has an order $o$, represented by a permutation of $\set{1\dots \ell}$ where $o_i$ represents the $i$-th fastest running dimension. We then define
    \[
        \id_R^o = \id_{r_{o_1}}^{\reg, o_1} \times \dots \times \id_{r_{o_\ell}}^{\reg, o_\ell}
    \]
    and $\id_T^o, \id_W^o$ similarly. Consider also the permutation of the dimensions by the order $o$
    \[
        \deffun{\sigma_o : \F^{d_1} \times \dots \times \F^{d_\ell} -> \F^{d_{o_1}} \times \dots \times \F^{d_{o_\ell}};}.
    \]

    Finally, with all this notation in place, the linear layout associated to this blocked layout is given by
    \[
        \deffun{\sigma_o^{-1} \circ \pa{\id^o_R \times \id^o_T \times \id^o_W} : \F^{\card{R}} \times \F^{\card{T}} \times \F^{\card{W}} -> \F^{d_1} \times \dots \times \F^{d_\ell};}.
    \]
    Note this is a linear map, as it is a composition of linear maps.
\end{proof}

\begin{proposition}
    The input and output layouts of \code{mma} and \code{wgmma} are linear layouts.
\end{proposition}
\begin{proof}
    In this case the logical matrix is two-dimensional. The definition of the tile is rather straightforward. For an input of bitwidth $b$, the lhs input and the output tile on registers for \code{mma} is given by
    \[
        \id^{\reg,1}_{\log_2(32/b)} \times
        \id^{\thread,1}_2 \times
        \id^{\thread,0}_3 \times
        \id^{\reg,0}_1 \times
        \id^{\reg,1}_1.
    \]
    and the rhs one by
    \[
        \id^{\reg,0}_{\log_2(32/b)} \times
        \id^{\thread,0}_2 \times
        \id^{\thread,1}_3 \times
        \id^{\reg,1}_1.
    \]
    which is the transpose of the first one with half the registers per thread.

    The input tile for the lhs of \code{wgmma} is given by multiplying the lhs tile of \code{mma} by $\id^{\warp,0}_2$ to cover the whole warp-group.

    The rest of the tile for the output is given by multiplying the first tile by $\id_W^o$, as defined in the proof of~\cref{prop:blocked} for a fixed order $o$---the order may be chosen by the implementation.

    The input warp part of the input tiles is then computed by looking at the warp that owns each output tile and making sure the given warp (resp.\  warp-group) has all the elements necessary to compute iteratively the reduction along the inner dimension. In other words, following the same warp order as the output, we need to broadcast (\ie, add a column of all zeros to the matrix) for every warp owning data on the inner dimension and multiply by the identity if it is the outer one.
\end{proof}

\begin{theorem}[Triton's Layout Engine]
    Consider the shape operations in Triton: \code{tt.trans}, \code{tt.reshape}, \code{tt.join}, \code{tt.split}, \code{tt.expand\_dims}, and \code{tt.broadcast}. The family of distributed layouts, as defined in~\cref{thm:distr}, is forward (resp.\ backward) closed under these operations. This means that for every input (resp.\ output in the image) distributed layout, there exists an output (resp.\ input) layout from the same family such that the operation effectively becomes a no-op. Furthermore, the family of distributed layouts is the smallest family of layouts satisfying this property.
\end{theorem}
\begin{proof}
    All these operations acting on the logical tensor are clearly linear, so the first part of the theorem follows naturally. Constructing the backward transfer function is essentially equivalent to constructing the forward ones.

    To prove the second part, we can reshape any tensor into the form $2\times2\times \dots \times2$ and apply dimension transpositions, reducing the problem to whether these operations can generate an arbitrary layout with zeros and ones over this hypercube. Since a layout of all ones can be created using the blocked encoding, and arbitrary zeros can be inserted by reducing along arbitrary dimensions, we do need all the linear layouts included in~\cref{thm:distr}, so this set is minimal.
\end{proof}

\subsection{Optimal Swizzling}\label{sec:opt_swiz}
In this section, we cover in detail the swizzling algorithm presented in the main text.

This algorithm computes an optimal swizzled layout that maximizes read/write vectorization while minimizing bank conflicts for arbitrary linear layouts. It is not difficult to generalize it to leverage \code{ldmatrix} and \code{stmatrix} and other intrinsics, but here, we will focus on vectorization for simplicity.

\textbf{Modeling Bank Conflicts in Linear Algebra}.
To model bank conflicts, we first define the vectorization set $V$ of size $2^v$ by choosing bases of $A_\reg \cap B_\reg$ as done for warp shuffles. For a data type with byte width $w$, let $b$ be the logarithm of the number of vectorized elements needed to cover all the shared memory banks. On modern GPUs, this is $b = \log_2 \tfrac{128}{2^vw}$.

We represent shared memory as a map
\[
    \deffun{S : \F^v \times \F^b \times \F^\ell -> \F^d;},
\]
where $\ell = d - v - b$. Here, the first space represents the vectorization $\vect$, the second represents the $\bank$, and the third represents the bank $\idx$ in shared memory.

By linearity, we obtain the following criterion for bank conflict-free memory access:

\begin{lemma}
    Given a shared memory layout $\deffun{S : \F^v \times \F^b \times \F^\ell -> \F^d;}$ and a distributed layout $L$ both representing elements of byte width $w$. Denote
    \[
    c = \card{\sspan\pa{S_\vect \cup S_\idx} \cap \sspan \pa{L_\thread}}.
    \]
    The memory operation will be performed in at least $c$ wavefronts.
    Even more, if each vectorized element covers $n \geq 1$ banks, \ie,  $n = \tfrac{2^vw}{4} \geq 1$, the operation will be performed in exactly $nc$ wavefronts.
\end{lemma}
\begin{proof}
    $S_\vect \subseteq L_\reg$, so its intersection with $L_\thread$ is trivial. It is then enough to look at $S_\idx \cap L_\thread$.
    We split the proof into three cases:

    \textbf{Each thread covers exactly one bank: $2^vw = 4$}. Since $\log_2 c = S_\idx \cap L_\thread$, there are $\log_2 c$ elements that will conflict performing the memory op in the bank with idx $0$. The same will happen with the other banks, so there will be exactly $c$ wavefronts, or $c-1$ bank conflicts.

    \textbf{Vectorized case. Each thread covers more than one bank: $n > 1$}. In this case, we have that $\card{S_\bank} = \tfrac{5}{\log_2 n}$. This corresponds to the case where we perform vectorized loads and stores. In current NVIDIA and AMD GPUs $n$ is allowed to be $2$ or $4$. In this case, the same reasoning as before goes through. We get $nc$ wavefronts because each vectorized shared memory operation is split into $128$ byte transactions.

    \textbf{Not enough vectorization. Each thread does not cover one full bank: $2^vw < 4$}. In this case, we do not have enough vectorization to cover one full bank with a thread, so there may be more bank conflicts on bank $0$ (and other banks) so we get that the number of wavefronts may be larger than $c$. Padding helps improve performance in this case at the expense of a higher memory footprint.
\end{proof}

When the vectorized elements cover at least one bank, and the intersection is trivial, the operation will have optimal throughput.

\textbf{Choosing a Basis for Bank Indices}.
Since we care about bank conflicts on reads and writes, we define
\[
    P = \sspan\pa{S_{\vect} \cup A_{\thread}} \cup \sspan\pa{S_{\vect} \cup B_{\thread}}.
\]
Note that $P$ is a union of two subspaces, so it is not a subspace itself.
As such, to minimize bank conflicts, we are interested in finding the largest basis $H$---and thus, the largest subspace---such that $P \cap \sspan \pa{H} = \set{0}$.

We start by constructing a basis $C$ of the complement subspace of $P$, \ie, we complete a basis of $\sspan \pa{P}$ into a basis of $\F^d$. It's clear that $\sspan \pa{P} \cap \sspan \pa{C} = \set{0}$.

Next, define the bases (\ie, the sets without the zero vector)
\[
E = A_\thread \backslash B_\thread, \quad F = B_\thread \backslash A_\thread.
\]
Without loss of generality, assume that $\card{E} \leq \card{F}$. We then enumerate their elements and construct
\[
    G = \set{e_i \oplus f_i | e_i \in E, f_i \in F, 1 \leq i \leq \card{E}}.
\]
By construction, $\sspan\pa{G}$ is in the complement of $P$. Even more, $\sspan\pa{G} \cap \sspan \pa{P} = \set{0}$.

Now, we determine the columns of $S_\idx$ as follows:
\begin{itemize}
\item If $\card{G} + \card{C} \geq \ell$, we select $\ell$ elements from $G \cup C$.
\item If $\card{G} + \card{C} < \ell$, bank conflicts are unavoidable. We add the remaining $\ell - \card{G} - \card{C}$ vectors from $A_\thread$, introducing both read and write bank conflicts.
\end{itemize}

Finally, having defined $S_\idx$, we determine $S_\bank$ by computing a basis for the complement of $\sspan\pa{S_\vect \cup S_\idx}$.

Let us now prove that this algorithm is indeed optimal. Before doing so, we will prove an abstract lemma from which the result will follow. We denote the cross product $U \times V$ as $U \oplus V$ as it makes the notation much clearer.
\begin{lemma}\label{lemma:abstract_swizzling}
    Given $U, V \subseteq \F^d$ subspaces. The largest subspace with trivial intersection with $U \cup V$ has dimension $d - \max\pa{\dim U, \dim V}$.
\end{lemma}
\begin{proof}
    Define $I = U \cap V$ and decompose $U = I \oplus E$, $V = I \oplus F$ where $E, F$ are the complementary spaces of $I$. Now extend $\sspan\pa{U \cup V}$ into the whole space via $C$ finding the decomposition
    \[
        \F^d = I \oplus E \oplus F \oplus C.
    \]
    In other words, any element of $\F^d$ is of the form $i \oplus e \oplus f \oplus c$ with $i \in I, e \in E, f \in F, c \in C$.

    Without loss of generality, consider $\dim U \leq \dim V$. Choose bases on $E$ and $F$ $\mathcal{B}_E = \set{e_1, \dots, e_k}, \mathcal{B}_F = \set{f_1, \dots, f_{k + n}}$ for $n \geq 0$ and define
    \[
        G = \sspan\set{e_i \oplus f_i | 1 \leq i \leq k}.
    \]
    More abstractly, $G$ can be defined via any injective linear map $\deffun{\phi : E -> F;}$ as $E \oplus \phi(E)$.

    Now, the set $C \oplus G$ has trivial intersection with $U \cup V$ and has dimension $d - \max\pa{\dim U, \dim V}$.

    It is also clear that this set is maximal, as a set of dimension $d - \dim V + 1$ would have non-trivial intersection with $V$.
\end{proof}

The correctness lemma is a corollary of the abstract lemma we just proved.
\begin{lemma}
    With notation as defined in~\cref{sec:conversion}, $\sspan\pa{S_\idx}$ is a subspace of dimension $\ell$ with minimal intersection with $P$.
\end{lemma}
\begin{proof}
    It follows from~\cref{lemma:abstract_swizzling} as $\sspan\pa{S_\idx}$ is defined as the subspace $C \oplus G$ in the proof of that theorem, which we have shown is maximal.
\end{proof}

\end{document}